\documentclass[12pt]{article}
\usepackage[letterpaper,margin=1in]{geometry}
\usepackage{amsmath,amssymb,amsthm,url}
\usepackage[pdftex]{color}
\usepackage{enumerate,bbm,fullpage}
\usepackage{lineno}
\usepackage{mathtools}
\usepackage{graphicx}
\usepackage{hyperref}
\usepackage{subcaption}
\usepackage[ruled,vlined]{algorithm2e}

\numberwithin{equation}{section}

\newcommand{\MAT}{\left[ \begin{array}}  
	\newcommand{\mat}{\end{array} \right]}

\newtheorem{theorem}{Theorem}[section]

\newtheorem{corollary}[theorem]{Corollary}
\newtheorem{lemma}[theorem]{Lemma}

\theoremstyle{definition}
\newtheorem{definition}[theorem]{Definition}
\newtheorem*{definition*}{Definition}
\newtheorem*{proposition*}{Proposition}
\newtheorem*{theorem*}{Theorem}
\newtheorem*{corollary*}{Corollary}
\newtheorem*{example*}{Example}
\newtheorem*{problem*}{Problem}

\newtheorem{runex}[theorem]{Example}

\theoremstyle{remark}
\newtheorem{remark}[theorem]{Remark}

\def\C{{\mathbbm C}}
\def\R{{\mathbbm R}}

\def\A{A}
\def\B{B}
\def\Atilde{\widetilde{A}}
\def\D{D}
\def\lambdatilde{\tilde{\lambda}}
\def\lambdahat{\widehat{\lambda}}

\def\M{M}
\def\m{m}

\def\N{N}
\def\Q{Q}
\def\r{r}
\def\R{R}
\def\Sigmatilde{\widetilde{\Sigma}}

\def\u{{\bf u}}
\def\uhat{{\bf \widehat{\u}}}
\def\utilde{{\bf \widetilde{\u}}}
\def\U{U}
\def\Uhat{\widehat{\U}}
\def\Utilde{\widetilde{\U}}
\def\v{{\bf v}}

\def\V{V}

\def\W{W}
\def\R{R}

\newcommand{\interior}[1]{%
	{\kern0pt#1}^{\mathrm{o}}%
}

\title{Fast One-Pass Sparse Approximation of the Top Eigenvectors of Huge Approximately Low-Rank Matrices? Yes, $MAM^*$!}

\author{
  Edem Boahen$^1$, Simone Brugiapaglia$^2$, Hung-Hsu Chou$^3$,\\ Mark Iwen$^{4}$, Felix Krahmer$^{567}$\\
   \\
    { \small $^1$Department of Mathematics, Michigan State University, \textsc{boahened@msu.edu}.}\\
{ \small $^2$Department of Mathematics and Statistics, Concordia University,} \\{\small \textsc{simone.brugiapaglia@concordia.ca}.}\\
{ \small $^3$Department of Mathematics, University of Pittsburgh, \textsc{edc93@pitt.edu}.}\\
 {\small $^4$Department of Mathematics, and Department of Computational Mathematics, Science}\\ {\small and Engineering (CMSE), Michigan State University, \textsc{iwenmark@msu.edu}.}\\
{\small $^5$Department of Electrical Engineering and Information Technology,}\\ {\small Technical University of Darmstadt, \textsc{felix.krahmer@tu-darmstadt.de}.}\\
{\small $^6$Department of Mathematics, \ Technical University of Munich }\\
{\small $^7$Munich Center for Machine Learning (MCML)}
}

\begin{document}

\maketitle

\begin{abstract}
Motivated by applications such as sparse PCA, in this paper we present provably-accurate one-pass
algorithms for the sparse approximation of the top eigenvectors of extremely massive matrices based on a single compact linear sketch. The resulting compressive-sensing-based approaches can approximate the leading eigenvectors of huge approximately low-rank matrices that are too large to store in memory based on a single pass over its entries while utilizing a total memory footprint on the order of the much smaller desired sparse eigenvector approximations. Finally, the compressive
sensing recovery algorithm itself (which takes the gathered compressive matrix measurements as input, and then
outputs sparse approximations of its top eigenvectors) can also be formulated to run in a time which principally depends on the size of
the sought sparse approximations, making its runtime sublinear in the size of the large matrix whose eigenvectors
one aims to approximate.  Preliminary experiments on huge matrices having $\sim 10^{16}$ entries illustrate the developed theory and demonstrate the practical potential of the proposed approach.

\end{abstract}

\section{Introduction}

Computing eigenvalues and eigenvectors of linear operators is fundamental to data science and computational mathematics. Beyond the ubiquitous example of Principal Component Analysis (PCA), its applications range from spectral clustering to spectral methods for Partial Differential Equations (PDEs). Traditional methods, such as the Power Iteration and the Lanczos Algorithm, often struggle with high computational costs and slow convergence rates, particularly when dealing with extremely large and nearly low-rank matrices \cite{golub1996matrix}. Thanks to advances in randomized numerical linear algebra \cite{martinsson2020randomized}, sketching and streaming approaches have recently been developed to address these challenges \cite{clarkson2009numerical,muthukrishnan2005data,woodruff2014sketching}. Motivated by this recent progress, we propose the $MAM^*$ method, which combines classical linear algebra results \cite{horn1991topics}, existing sketching approaches for approximating eigenvalues \cite{andoni2013eigenvalues,swartworth2023optimal}, and well-established, fast compressive sensing reconstruction techniques \cite{needell2009cosamp, BaileyIwenSpencer2012} to efficiently compute approximate leading eigenvectors of extremely massive matrices $A$. Our method is useful in a streaming setting where accessing $A \in \mathbbm{C}^{N \times N}$ from memory more than once is prohibitive due to its large size. In our case, we only need to access it once to compute a reduced sketch matrix $M A M^* \in \mathbbm{R}^{m \times m}$, where $m\ll N$, and $M \in \mathbbm{C}^{m \times N}$ is highly structured/memory efficient.

Before introducing the method and outlining the main contributions of the paper, we will briefly review recent progress in related research areas and discuss applications that motivate the development of the methods presented herein.

\subsection{Literature Review and Motivation}

The first research area strictly related to $MAM^*$-type sketching is the design of streaming algorithms for the fast approximation of eigenvalues and eigenvectors of large matrices. The early contribution \cite{andoni2013eigenvalues}, following up on an open problem formulated in \cite[\S 7.10.1]{muthukrishnan2005data}, proposed a method to estimate the top $k$ eigenvalues in the streaming model using $\mathcal{O}(k^2)$ space.  A significant contribution to this field is the $GAG^*$ method \cite{andoni2013eigenvalues, swartworth2023optimal}, able to approximate the eigenvalues of a symmetric matrix $A$ with high probability and up to an additive error $\epsilon \|A\|_F$, where $M=G$ is a random Gaussian matrix with $k = \mathcal{O}(1/\epsilon^2)$ rows.  A similar idea was also used in \cite{needell2022testing} to develop a fast sketching method for testing whether a symmetric matrix is positive semidefinite. 

Turning to eigenvectors, recent work has also focused on streaming approaches for the fast approximation of the top eigenvector of  large matrices.  The method proposed in \cite{kacham2024approximating} produces an approximation to the top eigenvector of $A^\top A$ assuming that the rows of a matrix $A \in \mathbbm{R}^{N \times n}$ are given in a streaming fashion using $\mathcal{O}(h \cdot n \cdot \text{polylog}(n))$ bits of memory, where $h$ is the number of heavy rows of the matrix (i.e., having Euclidean norm at least $\|A\|_F/\sqrt{n \cdot \text{polylog}(n)}$) and assuming that the ratio of the first two eigenvalues of $A^\top A$ is $\Omega(1)$.  It improves on previous work \cite{price2024spectral} under the extra assumption of uniform random ordering of the streamed rows.  Herein we show, alternatively, that a $MAM^*$-type sketch can provide accurate sparse approximations of multiple eigenvectors by simply relying on new eigenvector perturbation results combined with existing compressive sensing theory.

Other recent related work on the sketched Rayleigh–Ritz method \cite{nakatsukasa2024fast} shows that sketching methods can substantially reduce the computational cost of the Rayleigh-Ritz method \cite{saad2011numerical} for the fast approximation of eigenvalues of a matrix with respect to a given reduced basis.  For an $N \times N$ matrix and a basis of size $d \ll N$, sketching can reduce the cost of Rayleigh–Ritz from $\mathcal{O}(Nd^2)$ to $\mathcal{O}(d^3 + N d \log(d))$ operations.  We note that the reduced basis needed by the method in \cite{nakatsukasa2024fast} could be provided by the sparse approximate eigenvectors computed by, e.g., the $MAM^*$ method proposed herein when $A$ is extremely large. 

Other related sketching methods include older randomized methods for computing the Singular Value Decomposition (SVD) of a matrix. The randomized method in \cite{halko2011finding} is able to, e.g., compute the $k$ dominant components of the SVD of an $m \times n$ dense matrix in $\mathcal{O}(mn \log(k))$ operations, improving on the $\mathcal{O}(mn k)$ cost of classical algorithms. In addition, for matrices that are too large to store in memory, the randomized SVD method in \cite{halko2011finding} requires a constant number of passes over the data as opposed to $\mathcal{O}(k)$ passes of classical algorithms.  For comprehensive treatment of randomized SVD and other randomized numerical linear algebra methods we refer the reader to the survey \cite{martinsson2020randomized}.  
This paper builds upon the core ideas developed in these papers
to strengthen the method for matrices $A$ with top eigenvectors that are sparse or compressible. This setup allows for a method with a number of advantages. Namely, 
the methods developed herein  are strictly one-pass and utilize different (and more compact) $MAM^*$-type sketches.

The sparsity assumption for the top eigenvector underlying these improvement is motivated by the observation that a smaller number of features corresponds to better interpretability of the model.
This creates a close conceptual tie
to \emph{sparse PCA}, the problem of representing a matrix as a sum of sparse rank-$1$-matrices. 
After its introduction in \cite{zou2006sparse}, in little more than a decade sparse PCA has become an essential asset in the modern data scientist's toolkit thanks to its ability to enhance the interpretability of the standard PCA paradigm. For further details on sparse PCA and historical remarks, we refer the interested reader to \cite[Chapter~8]{wainwright2019high}.  
One of the earliest sparse PCA techniques was based on the \emph{truncated power method}, proposed in \cite{yuan2013truncated}. The idea is to alternate the classical power method iteration with a truncation step, where only the $k$-largest absolute entries of the approximate eigenvector are kept and rescaled each iteration so as to achieve normalization. The truncated power method was later extended to the truncated Rayleigh flow method (also known as ``rifle'') in order to handle generalized eigenvalue problems in \cite{tan2018sparse}. 
For the purposes of this paper we simply remark that the proposed $MAM^*$-method developed herein can also be considered as the first known one-pass algorithm for sparse PCA of extremely massive matrices $A$.

\subsection{The $MAM^*$ Method}

To get an idea for how the proposed $MAM^*$ approach works, consider a very simple rank one $N \times N$ Hermitian matrix $A = {\bf u}{\bf u}^*$ where ${\bf u} \in \mathbbm{C}^N$ is compressible (i.e., approximately $s$-sparse).  Letting $E_{j,k} \in \mathbbm{C}^{N \times N}$ be the matrix of all zeros with a single $1$ in it's $(j,k)^{\rm th}$-entry, and $M \in \mathbbm{C}^{m \times N}$ be a compressive sensing matrix having an associated fast reconstruction algorithm (see, e.g.,  \cite{iwen2014}), we can compute $M A M^* = \sum_{j,k} a_{j,k} M E_{j,k} M^* = M{\bf u} (M {\bf u})^* \in \mathbbm{C}^{m \times m}$ by streaming over the entries of $A$ just once.  We can then simply $(i)$ compute the top eigenvector, $M{\bf u}$, of our small sketched $m \times m$ matrix $M A M^*$, and then $(ii)$ use it's eigenvector in our fast compressive sensing algorithm to rapidly approximate ${\bf u}$.  Of course, when $A$ isn't such a simple rank one matrix the situation becomes more complicated.  Nonetheless, as we shall see it is possible to adapt this basic idea to still work for a large class of nearly rank $r \ll N$ matrices.  Extension to the Singular Value Decomposition (SVD) of non-Hermitian matrices is also possible.

To see why this approach should still work for higher-rank, e.g., Hermitian and Positive SemiDefinite (PSD) matrices $A$ we can consider $M A M^* \in \mathbbm{C}^{m \times m}$, $$M A M^* = \sum_{j=1}^m \tilde{\lambda}_j \cdot { \tilde{\bf u}}_j {\tilde{ \bf u}}_j^*,$$ where $\tilde{\lambda}_1 \geq \tilde{\lambda}_2 \geq \dots \geq \tilde{\lambda}_m \geq 0$, together with an eigendecomposition of $A \in \mathbbm{C}^{N \times N}$, $$A = \sum_{j=1}^N \lambda_j \cdot {\bf u}_j {\bf u}_j^*,$$ where $\lambda_1 \geq \lambda_2 \geq \dots \geq \lambda_N \geq 0$.  If $A$ is nearly rank $r \leq m \ll N$ then we can see that 
\begin{equation}
\sum_{j=1}^m \tilde{\lambda}_j \cdot { \tilde{\bf u}}_j {\tilde{ \bf u}}_j^* ~=~ M A M^* ~=~ \sum_{j=1}^N \lambda_j \cdot (M {\bf u}_j) (M {\bf u}_j)^* ~\approx~ \sum_{j=1}^m \lambda_j \cdot (M {\bf u}_j) (M {\bf u}_j)^*.
\label{equ:BasicCSemasIdea}
\end{equation}

Note now that if each element of the orthonormal subset of eigenvectors of $A$, $\{ {\bf u}_j \}_{j 
\in [m]} \subset \mathbbm{C}^N$, in \eqref{equ:BasicCSemasIdea} is also sparse/compressible, and if $M \in \mathbbm{C}^{m \times N}$ has the Restricted Isometry Property (RIP) \cite{FoucartRauhut2013}, then the set of ``\emph{almost-eigen}"vectors of $MAM^*$, $\{ M{\bf u}_j \}_{j 
\in [m]} \subset \mathbbm{C}^m$, should be also be \emph{nearly orthonormal}.  As a result, in this case it's reasonable to expect that something akin to a Davis–Kahan $\sin\theta$ theorem \cite{Yusinatheta2015} should guarantee that ${ \tilde{\bf u}}_j \approx M{\bf u}_j$ for all $j \in [m]$ as long as the associate eigenvalue gaps $|\lambda_{j+1} - \lambda_j|$ are reasonably large for all $j \in [m]$.  We prove herein that this is indeed the case by developing theory bounding, among other quantities, the \emph{measurement error} $\min_{\phi \in [0,2\pi)} \|  \utilde_j - M ( \mathbbm{e}^{\mathbbm{i} \phi} \u_j) \|_2$.  Finally, if $M$ has, e.g., an RIP property, and if this measurement error is small, then it should be possible to run a fast compressive sensing algorithm with $\utilde_j \in \mathbbm{C}^m$ as input in order to quickly approximate $\u_j \in \mathbbm{C}^N$ up to a global phase.  This is the general eigenvector approximation strategy proposed and analyzed herein.  See Algorithm~\ref{Alg:StreamOverA} for pseudo-code outlining two strategies for computing the $MAM^*$ sketch in one low-memory pass over the entries of $A$, and 
Algorithm~\ref{Alg:ApproxfromMAM} for pseudo-code outlining the approximation of the eigenvectors of $A$ from a previously computed $MAM^*$ sketch.

\begin{algorithm}[htbp]
    \caption{Computing $MAM^* \in \mathbbm{C}^{m \times m}$ in One Pass} \label{Alg:StreamOverA}

    \DontPrintSemicolon
    \SetAlgoVlined
    \SetKwFor{For}{for}{}{}
    \SetKwFor{While}{while}{}{}

    \KwIn{$(i)$ Pointer to (an algorithm for generating) $A \in \mathbbm{C}^{N \times N}$ (e.g., by streaming though its entries $a_{k,j} \in \mathbbm{C}$, or columns ${\bf a}_{:,j} \in \mathbbm{C}^N$), and a $(ii)$ pointer to an algorithm for either computing $M{\bf v} \in \mathbbm{C}^{m}$ for any $\mathbf{v} \in \mathbbm{C}^N$, or for generating the columns ${\bf m}_{:,j} \in \mathbbm{C}^N$ of $M$ for all $j \in [N]$.}
    
    \KwOut{$MAM^* \in \mathbbm{C}^{m \times m}$}

\vspace{.2in}

     \tcc{Computing $MAM^*$ in One Pass Over the Entries of $A \in \mathbbm{C}^{N \times N}$}
     Initialize $Q \in \mathbbm{C}^{m \times m} \leftarrow $ zero matrix\;
    \For{$j,k \in [N]$}{
        $Q \leftarrow Q + a_{j,k} {\bf m}_{:,j}{\bf m}_{:,k}^*$\;  }
    Return $Q$\;

\vspace{.2in}

     \tcc{Computing $MAM^*$ in One Pass Over the Columns of $A \in \mathbbm{C}^{N \times N}$}
          Initialize $R \in \mathbbm{C}^{m \times m} \leftarrow $ zero matrix\;
    \For{$j\in [N]$}{
        \tcp{Compute $MA$ column by column (possibly also in one pass)}
        ${\bf q} \leftarrow M {\bf a}_{:,j}$\; 
    \tcp{Compute Rank One $MAM^*$ Update}
    $R \leftarrow R + {\bf q}{\bf m}_{:,j}^*$\;}
    Return $R$\;
\end{algorithm}

\begin{algorithm}[htbp]
    \caption{Approximating $\{ \u_j \}_{j \in [\ell]}$ of Hermitian $A \in \mathbbm{C}^{N \times N}$ via $MAM^* \in \mathbbm{C}^{m \times m}$} \label{Alg:ApproxfromMAM}

    \DontPrintSemicolon
    \SetAlgoVlined
    \SetKwFor{For}{for}{}{}
    \SetKwFor{While}{while}{}{}

    \KwIn{$MAM^* \in \mathbbm{C}^{m \times m}$, $\ell \in [m]$, pointer to compressive sensing algorithm $\mathcal{A}: \mathbbm{C}^m \rightarrow \mathbbm{C}^N$}
    
    \KwOut{$s$-sparse $\u_j' \approx$ eigenvector $\u_j$ of $A$, $\forall j \in [\ell]$}

\vspace{.2in}

    \tcc{Compute the Top-$\ell$ Eigenvectors of $MAM^*$}
    $\{ \utilde_j \}_{j \in [\ell]} \leftarrow$ top-$\ell$ eigenvectors of $MAM^* \in \mathbbm{C}^{m \times m}$\; 

    \vspace{.1 in}
    \tcc{Use eigenvectors of $MAM^*$ as compressive sensing measurements}
    \For{$j \in [\ell]$}{
    $\u_j' \leftarrow \mathcal{A}\left( \utilde_j \right)$\; }

    Return $\{\u_j'\}_{j \in [\ell]}$\;

\end{algorithm}

Fixing notation, let $\A\in\mathbbm{C}^{\N\times \N}$ be a large Hermitian and PSD matrix and $\M\in\mathbbm{C}^{\m\times \N}$ be a measurement matrix with $\m < \N$.  Set $\Atilde=\M\A\M^*$ to be the sketch matrix of (much) smaller size. Our general objective is to use eigenvectors of $\Atilde$ to gain information about the eigenvectors of $A$. Formally speaking, if
\begin{align} \label{def:SVD_A}
    &\A =U\Sigma U^*= \sum_{j=1}^\N \lambda_j\cdot\u_j\u_j^*, \\ \label{def:SVD_A_Atilde} 
    &\Atilde =\Utilde\Sigmatilde \Utilde^*=\sum_{j=1}^\m \lambdatilde_j\cdot\utilde_j\utilde_j^*
\end{align}
are eigendecompositions of $\A$ and $\Atilde$, respectively, then we aim to recover $\u_j$ based on $\utilde_j$. Here we denote the nonnegative eigenvalues by $\lambda_j \geq 0$ and $\lambdatilde_j  \geq 0$.  

Note that by construction we can also express $\Atilde$ as
\begin{align*}
    \Atilde
    =\M\A\M^*
    &= \sum_{j=1}^N \lambda_j\cdot(\M\u_j)(\M \u_j)^*\\
    &= \sum_{j=1}^N \lambda_j\|\M\u_j\|^2_2\cdot\frac{\M\u_j}{\|M\u_j\|_2}\frac{(\M\u_j)^*}{\|\M\u_j\|_2}\\
    &= \sum_{j=1}^N \lambdahat_j\cdot\uhat_j\uhat_j^*,
\end{align*}
where
\begin{align}
&\lambdahat_j := \lambda_j\|\M\u_j\|_2^2 \quad {\rm and} \quad \uhat_j := \frac{\M\u_j}{\|M\u_j\|_2}. \label{def:lambda}
\end{align}
Again, the the first crucial observation here is that a generalized eigenvector perturbation result proven herein guarantees that \eqref{equ:BasicCSemasIdea} holding also implies that ${ \tilde{\bf u}}_j \approx M{\bf u}_j$ holds in a wide variety of settings.  The second crucial observation is that a wide variety of highly structured measurement matrices $M$ exist which not only guarantee that ${ \tilde{\bf u}}_j \approx M{\bf u}_j$ holds, but that also allow fast (e.g., sublinear-time) compressive sensing algorithms to recover ${\bf u}_j$ when given ${ \tilde{\bf u}}_j \approx M{\bf u}_j$ as input.  Examples of the type of results this theoretical framework allows one to prove follow.


\subsection{Main Results and Contributions}
Before stating example results we fix some notation used
throughout.  Given $A \in \mathbbm{C}^{N \times N}$ with SVD
$A = \sum_{j=1}^N \sigma_j(A) \u_j \v_j^*$, we write
$A_{{\setminus}r} := \sum_{j=r+1}^N \sigma_j(A) \u_j \v_j^*$ for the rank $N-r$ matrix obtained by excluding the top-$r$ singular vectors of $A$, and denote the
matrix nuclear norm of $A$ by $\|A\|_* := \sum_{j=1}^N \sigma_j(A)$.  For a vector
$\u \in \mathbbm{C}^N$ and $s \in [N]$, $(\u)_s$ denotes a best
$s$-sparse approximation to $\u$, i.e.,
$(\u)_s := \operatorname{arg\,min}_{\{ \v \in \mathbbm{C}^N \,|\, \| \v\|_0 \le s \} }
\| \u - \v \|_2$, where $\| \v\|_0 := |\{ j ~|~ v_j \neq 0 \} |$.\footnote{Note that the best $s$-term approximation $(\u)_s$ to $\u$ may not always be unique.  In such cases one may simply select the best $s$-term approximation whose support $\subset [N]$ is first in lexicographical order.}

The following example theorem is proven by combining the eigenvector embedding theory developed in Section~\ref{sec:EigenvectorMeasApprox} with the randomized numerical linear algebra and compressive sensing techniques discussed in sections~\ref{sec:Preliminaries} and \ref{sec:CompressiveSense}. It is proven in Section~\ref{sec:PROOFOFMAINRESULT_LINEAR}.

\begin{theorem} \label{THM:MAINRESULT_LINEARcosamp}
Let $q \in (0,1/3)$, $c \in [1,\infty)$, $\ell, r \in [N]$, and $\epsilon \in (0,1)$ be such that $\epsilon < \min \left\{ \frac{1}{20}, \frac{1}{4}\left(\frac{1-3q}{1+q}\right) \right\}$ and $2 \leq \ell \leq r$. Suppose that $\A\in \mathbbm{C}^{\N\times\N}$ is Hermitian and PSD with eigenvalues $\lambda_1 \geq \lambda_2 \geq \dots \geq \lambda_N \geq 0$ satisfying
\begin{enumerate}
    \item $\lambda_j = c q^j$ for all $j \in [\ell] \subseteq [r]$, and
    \item $\|\A_{\backslash \r}\|_* \leq \epsilon \lambda_\ell$.
\end{enumerate}
Choose $s\in [\N]$, $p, \eta \in (0,1)$, and form a random matrix $\M \in\mathbbm{C}^{\m\times\N}$ with $m = \mathcal{O}\left( \max \left\{s,\frac{r}{\epsilon^2} \right\} \log^4\left( N/p\epsilon^2 \right) \right)$ as per Theorem~\ref{Thm:MainCosampMatrixSetup}.  Let $\utilde_j$ and $\u_j$ be the ordered eigenvectors of $MAM^*$ \eqref{def:SVD_A_Atilde} and $A$ \eqref{def:SVD_A}, respectively, for all $j \in [\ell]$.  Then, there exists a compressive sensing algorithm $\mathcal{A}_{\rm lin}: \mathbbm{C}^m \rightarrow \mathbbm{C}^N$ and an absolute constant $c' \in \mathbbm{R}^+$ such that 
\begin{equation*}
		\min_{\phi \in [0,2\pi)} \left \| \mathbbm{e}^{\mathbbm{i} \phi} \u_j - \mathcal{A}_{\rm lin}(\utilde_j) \right \|_2 < c' \cdot \max \left\{ \eta, \frac{1}{\sqrt{s}} \| \u_j - (\u_j)_s\|_1 + \sqrt{\epsilon} \cdot q^{1-j} \right\}
\end{equation*}
holds for all $j \in [\ell]$ with probability at least $1-p$.\footnote{Given $\u \in \mathbbm{C}^N$ and $s \in [N] := \{ 1, 2, \dots, N\}$, the vector $\u_s \in \mathbbm{C}^N$ denotes a best possible $s$-sparse approximation to $\u$.  Let $\Sigma_s := \{ {\bf x} ~|~ \| {\bf x} \|_0 \leq s \} \subset \mathbbm{C}^N$.  Then, $\| \u - \u_s \|_p = \inf_{\v \in \Sigma_s} \| \u - \v \|_p$ holds for all $p \in [1,\infty]$.}  Furthermore, all $\ell$ estimates $\left\{ \mathcal{A}_{\rm lin}(\utilde_j) \right \}_{j \in [\ell]}$ can be computed in $\mathcal{O}\left(m^3 + \ell N \log N \cdot \log(1/\eta) \right)$-time from $MAM^* \in \mathbbm{C}^{m \times m}$ using $\mathcal{O}(N)$-memory.
\end{theorem}
\begin{remark}[On the constant $c'$ in Theorem~\ref{THM:MAINRESULT_LINEARcosamp}]
The constant $c'= 7c''$, where $7$ comes from the explicit bounds
in the proof of Theorem~\ref{MainTHM:ExpDecayingSingVlaues} and $c''$ is the absolute constant from~\cite[Theorem~A]{needell2009cosamp}.  
\end{remark}

Although it is possible to relax the exact spectral decay $\lambda_j = c q^j$ to, e.g., $\lambda_j \leq c q^j$ with a proof strategy similar to the one presented in this paper, such generalizations require more technical assumptions (such as minimal gaps between leading eigenvalues) which complicate the results without providing additional insight into the scaling of the approach.  Hence, such variants are hereby omitted (see  Section~\ref{sec:HOWTOGeneralizeMainThms} for additional related discussion).

In addition, the following additional example theorem is proven in Section~\ref{sec:PROOFOFMAINRESULT_SUBLINEAR}.  It differs from Theorem~\ref{THM:MAINRESULT_LINEARcosamp} only in the choice of the compressive sensing algorithm used to recover the top eigenvectors of $A$ from the eigenvectors of $MAM^*$.  We emphasize that many other similar results can also be proven using the framework herein by simply continuing to vary one's choice of this compressive sensing algorithm.

\begin{theorem} \label{THM:MAINRESULT_SUBLINEAR}
Let $q \in (0,1/3)$, $c \in [1,\infty)$, and $\ell, r, 1/\epsilon \in [N]$ be such that $\epsilon < \min \left\{ \frac{1}{20}, \frac{1}{4}\left(\frac{1-3q}{1+q}\right) \right\}$ and $2 \leq \ell \leq r$. Suppose that $\A\in \mathbbm{C}^{\N\times\N}$ is Hermitian and PSD with eigenvalues $\lambda_1 \geq \lambda_2 \geq \dots \geq \lambda_N \geq 0$ satisfying
\begin{enumerate}
    \item $\lambda_j = c q^j$ for all $j \in [\ell] \subseteq [r]$, and
    \item $\|\A_{\backslash \r}\|_* \leq \epsilon \lambda_\ell$.
\end{enumerate}
Choose $s\in [\N]$, $p \in (0,1)$, and form a random matrix $\M \in\mathbbm{C}^{\m\times\N}$ with $m = \mathcal{O}\left( \max \left\{ s^2, \frac{r^2}{\epsilon^2}  \right\} \log^5(N/p) \right)$ as per Theorem~\ref{Thm:MainSublinearMatrixSetup}.  Let $\utilde_j$ and $\u_j$ be the ordered eigenvectors of $MAM^*$ \eqref{def:SVD_A_Atilde} and $A$ \eqref{def:SVD_A}, respectively, for all $j \in [\ell]$.  Then, there exists a compressive sensing algorithm $\mathcal{A}_{\rm sub}: \mathbbm{C}^m \rightarrow \mathbbm{C}^N$ and $\beta^{A}_{M} \in \mathbbm{R}^+$ such that 
\begin{align*}
		&\min_{\phi \in [0,2\pi)} \left \| \mathbbm{e}^{\mathbbm{i} \phi} \u_j - \mathcal{A}_{\rm sub}(\utilde_j) \right \|_2\\ 
        & \qquad \qquad < \| \u_j - (\u_j)_{ 2s } \|_2 + 6(1+\sqrt{2}) \left( \frac{\| \u_j - (\u_j)_s \|_1}{\sqrt{ s }} + \beta^{A}_{M} \sqrt{\epsilon} \cdot q^{1-j}\right), 
\end{align*}
holds for all $j \in [\ell]$ with probability at least $1-p$.  Furthermore, all $\ell$ estimates $\left\{ \mathcal{A}_{\rm sub}(\utilde_j) \right \}_{j \in [\ell]}$ can be computed in just $\mathcal{O}\left(m^3 \right)$-time from $MAM^* \in \mathbbm{C}^{m \times m}$.
\end{theorem}

Note that the worst-case measurement bounds on $m$ provided by Theorem~\ref{THM:MAINRESULT_SUBLINEAR} are worse than those provided by Theorem~\ref{THM:MAINRESULT_LINEARcosamp} (e.g., they scale quadratically in $r$ and $s$ as opposed to linearly). In exchange, however, the overall runtime complexity guaranteed by Theorem~\ref{THM:MAINRESULT_SUBLINEAR} is significantly better than that provided by Theorem~\ref{THM:MAINRESULT_LINEARcosamp} for large $N$, scaling only \emph{sub-linearly in $N$} as opposed to linearly.  Meanwhile, the approximation error bounds provided by both results are comparable up to the numerical constant $\beta^{A}_{M}$ appearing in Theorem~\ref{THM:MAINRESULT_SUBLINEAR}.

Towards a better understanding of $\beta^{A}_{M}$, let $\phi'_j \in [0,2\pi)$ be such that $\left\| \utilde_j - \mathbbm{e}^{\mathbbm{i} \phi'_j} M \u_j \right\|_2 = \min_{\phi \in [0,2\pi)} \| \mathbbm{e}^{\mathbbm{i} \phi} \utilde_j - M \u_j \|_2$ for all $j \in [\ell]$.  Then 
$\beta^{A}_{M} = 7 \max_{j \in [\ell]} \beta_{m}(\utilde_j - \mathbbm{e}^{\mathbbm{i} \phi'_j}M\u_j)$, where $\beta_{m}(\cdot)$ is defined as per \eqref{equ:DefofBeta_mn}, can be expected to be $\mathcal{O}(1)$ as long as the error vectors $\utilde_j - \mathbbm{e}^{\mathbbm{i} \phi'_j}M\u_j \in \mathbbm{C}^m$ are ``flat'' with $\left\| \utilde_j - \mathbbm{e}^{\mathbbm{i} \phi'_j}M\u_j \right\|_\infty = \mathcal{O} \left( m^{-\frac{1}{2}}\left\| \utilde_j - \mathbbm{e}^{\mathbbm{i} \phi'_j}M\u_j \right \|_2 \right)$ for all $j \in [\ell]$ (see, e.g., Remark~\ref{remark:sublinearBetaissmall}).  Indeed, this is the case in all of the experiments we have run so far (see, e.g., Section~\ref{sec:BetaExperiments}).  As a consequence, we expect that the approximation errors bounded by Theorem~\ref{THM:MAINRESULT_LINEARcosamp} and~\ref{THM:MAINRESULT_SUBLINEAR} are comparable in most cases.  We leave an improved theoretical understanding of the numerical constant $\beta^{A}_{M}$ for future work.

\subsubsection{The General Proof Framework, and a Guide to Proving Other Variants}
\label{sec:HOWTOGeneralizeMainThms}

Both Theorems~\ref{THM:MAINRESULT_LINEARcosamp} and~\ref{THM:MAINRESULT_SUBLINEAR} are proven in three phases. First, in Section~\ref{sec:Preliminaries}, we prove that the top eigenvalues of $MAM^*$ are good approximations to the top eigenvalues of $A$ for a very general class of Johnson–Lindenstrauss (JL) embedding matrices $\M$ and matrices $\A$.  This generalizes prior work focused on the case where $M$ is a Gaussian random matrix \cite{andoni2013eigenvalues,swartworth2023optimal}.  Second, in Section~\ref{sec:EigenvectorMeasApprox}, we prove that the eigenvectors of $MAM^*$ provide accurate compressive measurements of the eigenvectors of Hermitian/symmetric and PSD matrices $A$.  The new bounds proven therein are reminiscent of classical eigenvector perturbation results such as the Davis–Kahan $\sin\theta$ theorem \cite{Yusinatheta2015}, except that they involve the non-classical decomposition on the right side of \eqref{equ:BasicCSemasIdea}.  Third, in Section~\ref{sec:CompressiveSense}, performance bounds are derived for specific compressive sensing algorithms when they take noisy measurements resulting from eigenvectors of $MAM^*$ as input.  Finally, Theorems~\ref{THM:MAINRESULT_LINEARcosamp} and~\ref{THM:MAINRESULT_SUBLINEAR} concerning matrices $A$ whose top eigenvalues decay exponentially are then proven in Section~\ref{sec:MainResFinalProofs} using the supporting theory developed in Sections~\ref{sec:Preliminaries}--\ref{sec:CompressiveSense}.  In addition, Section~\ref{sec:Experiments} provides preliminary experiments demonstrating the empirical performance of the strategy analyzed in Theorem~\ref{THM:MAINRESULT_SUBLINEAR} on large matrices $A$ having $\sim 10^{16}$ entries.\\

\noindent {\small \bf A Guide to Proving Alternate Variants of Theorems~\ref{THM:MAINRESULT_LINEARcosamp} and~\ref{THM:MAINRESULT_SUBLINEAR}}\\

\noindent 
 While the results and proofs in this paper pertain to matrices $A$ whose top eigenvalues decay exponentially, we emphasize here that alternate theorems for other classes of approximate low-rank matrices $A$ of interest can be proven in the same fashion using the theoretical framework developed above.  More precisely, let $\lambda_j, \lambdatilde_j, \lambdahat_j$ be defined as in \eqref{def:SVD_A}--\eqref{def:lambda}.  In Section~\ref{sec:Preliminaries} we begin by bounding $|\lambda_j - \lambdatilde_j|$ using
Theorem~\ref{theorem:MAM}.  
Although Theorem~\ref{theorem:MAM} provides weaker error bounds for some matrices $\A$ than \cite{swartworth2023optimal} does when $\M$ is a Gaussian random matrix, it also applies to \emph{all JL embedding matrices $M$} including, e.g., those with fast matrix-vector multiplication algorithms such as \cite{Ailon2008fast,ailon2009the,krahmer2011new,bourgain2015toward,iwen2024on}.  Furthermore, Theorem~\ref{theorem:MAM} will hold independently of any conditions placed on the spectral decay of $A$, and so can be used to prove Lemma~\ref{lem:distance_lambda_j_lambda_tilde_j} bounding $| \lambdatilde_j - \lambdahat_j |$ under very general conditions.  

In Section~\ref{sec:EigenvectorMeasApprox} we next bound the error 
$$\min_{\phi \in [0,2\pi)} \|  \utilde_j - M ( \mathbbm{e}^{\mathbbm{i} \phi} \u_j) \|_2 ~\lesssim~ \sqrt{2 - 2 \| M \u_j \|_2 ~| \langle \utilde_j, \uhat_j \rangle|}$$
by lower bounding $|\langle \utilde_j,\uhat_j \rangle|$ in Theorem~\ref{thm:angle_lower_bound}.  Although quite general, Theorem~\ref{thm:angle_lower_bound} involves several difficult-to-interpret quantities $g_j, \kappa_j, \nu_j$ defined in terms of the eigenvalues of $\Atilde$ (i.e., the $\lambdatilde_j$'s).  Given that we are primarily interested in understanding how the spectral properties of the original matrix $\A$ (as opposed to $\Atilde$) influence the behavior of our proposed algorithms herein, we simplify and recast Theorem~\ref{thm:angle_lower_bound} in terms of the eigenvalues of $\A$ in Theorem~\ref{MainTHM:ExpDecayingSingVlaues}.  The exponentially-decaying top-eigenvalue model for $A$ considered in Theorems~\ref{THM:MAINRESULT_LINEARcosamp} and~\ref{THM:MAINRESULT_SUBLINEAR} (i.e., $\lambda_j = c q^j$ for $q < 1/3$) ultimately originates in Theorem~\ref{MainTHM:ExpDecayingSingVlaues} as a way to help us prove simple upper bounds on the quantities $g_j$, $\kappa_j$, and $\nu_j$ from Theorem~\ref{thm:angle_lower_bound}.  These simple upper bounds on $g_j$, $\kappa_j$, and $\nu_j$ then provide clean bounds on the measurement errors $\min_{\phi \in [0,2\pi)} \|  \utilde_j - M ( \mathbbm{e}^{\mathbbm{i} \phi} \u_j) \|_2$ via Corollary~\ref{Coro:approxerrorinEigenvectors}.  That said, we hasten to point out that similar results certainly hold for significantly more general classes of matrices $\A$ than what we focus on in Theorems~\ref{MainTHM:ExpDecayingSingVlaues}, \ref{THM:MAINRESULT_LINEARcosamp}, and~\ref{THM:MAINRESULT_SUBLINEAR}.  Cleanly bounding the quantities $g_j$, $\kappa_j$, and  $\nu_j$ in Theorem~\ref{thm:angle_lower_bound} from above in terms of the eigenvalues of $A$ will be the principal task involved in proving them.

By the conclusion of Section~\ref{sec:EigenvectorMeasApprox} it's  established that ${ \tilde{\bf u}}_j \approx M{\bf u}_j$ holds as long as $\M$ satisfies the conditions of Theorem~\ref{MainTHM:ExpDecayingSingVlaues} (which match those of Theorem~\ref{theorem:MAM}).  At this point it is therefore clear that one's favorite Compressive Sensing (CS) algorithm $\mathcal{A}:  \mathbbm{C}^m \rightarrow \mathbbm{C}^N$ will have $\mathcal{A} \left({\tilde{\bf u}}_j\right) \approx {\bf u}_j$ whenever $\M$ \emph{both} satisfies the properties required by $\mathcal{A}$ \emph{and} satisfies the conditions of Theorem~\ref{theorem:MAM}.  In Section~\ref{sec:CompressiveSense} we finish building the  theoretical foundations needed to prove Theorems~\ref{THM:MAINRESULT_LINEARcosamp} and~\ref{THM:MAINRESULT_SUBLINEAR} by proving that, indeed, measurement matrices $\M$ exist which not only satisfy the conditions of Theorem~\ref{theorem:MAM}, but that also have the properties required by two important types of CS algorithms:  $(i)$ standard CS algorithms requiring the RIP (see Lemma~\ref{lem:MergedPropsCosamp}), and $(ii)$ sublinear-time CS algorithms requiring $\M$ to have highly-structured combinatorial properties (see Theorem~\ref{Thm:MainSublinearMatrixSetup}).  This theory, like that in Section~\ref{sec:Preliminaries}, will hold independently of any conditions placed on the spectral decay of $A$ and so can be used to prove alternate variants of Theorems~\ref{THM:MAINRESULT_LINEARcosamp} and~\ref{THM:MAINRESULT_SUBLINEAR} without modification.

We are now prepared to begin proving our main results.

\section{Eigenvalue Preliminaries}
\label{sec:Preliminaries}

The principal goal of this section is to upper bound the differences $|\lambda_j - \lambdatilde_j|$ and $|\lambdatilde_j - \lambdahat_j|$ as 
defined in \eqref{def:SVD_A} -- \eqref{def:lambda}.  In the course of doing so we will review some standard terminology and results related to fast Johnson–Lindenstrauss (JL) embeddings.
\begin{definition}[$\epsilon$-JL map] \label{JLmapdef}
    An $\epsilon$-JL map of a set $\mathcal{S} \subset \mathbbm{C}^N$ is a linear map $f:\mathbbm{C}^{\N}\to\mathbbm{C}^\m$ such that for all ${\bf x} \in \mathcal{S}$,
    \begin{equation}
    \label{eq:JLproperty}
        (1-\epsilon)\|{\bf x}\|_2^2
        \leq \|f({\bf x})\|_2^2
        \leq (1+\epsilon)\|{\bf x}\|_2^2.
    \end{equation}
    For notation simplicity, we say a matrix $\M \in \mathbbm{R}^{m \times N}$ is an $\epsilon$-JL map if $f({\bf x}):=\M {\bf x}$ is an $\epsilon$-JL map of given set.
\end{definition}
Note that since  $\epsilon\in(0,1)$, $1-\epsilon\leq \sqrt{1-\epsilon}$ and $\sqrt{1+\epsilon} \leq 1+\epsilon$.  Hence,
\begin{equation*}
     \sqrt{1-\epsilon}~\|{\bf x}\|_2
    \leq \|\M{\bf x}\|_2
    \leq \sqrt{1+\epsilon}~\|{\bf x}\|_2. 
\end{equation*}
implies that
\begin{equation} \label{equ:droppedsqaresJLdef}
    (1-\epsilon)\|{\bf x}\|_2 \leq \|\M{\bf x}\|_2 \leq (1+\epsilon)\|{\bf x}\|_2.
\end{equation}
As a consequence, any matrix $\M \in \mathbbm{R}^{m \times N}$ satisfying Definition~\ref{JLmapdef} also satisfies \eqref{equ:droppedsqaresJLdef}.  

We also recall the definition of Restricted Isometry Property (RIP), used later in our analysis.
\begin{definition}[RIP of order $(s,\epsilon)$] \label{def:RIP}
For ${\bf x} \in \mathbbm{C}^N$, let $\| {\bf x} \|_0 := |\operatorname{supp}({\bf x})|$
denote the number of nonzero entries of {\bf x}, and set
$\Sigma_s := \{ {\bf x} \mid \|{\bf x}\|_0 \le s\} \subset \mathbbm{C}^N$.
A matrix $M \in \mathbb{R}^{m \times N}$ is said to have the
\emph{Restricted Isometry Property} (RIP) of order $(s,\varepsilon)$
for some $0 \le s \le N$ and $\varepsilon \in (0,1)$ if
$f({\bf x}) := M{\bf x}$ satisfies \eqref{eq:JLproperty} for all ${\bf x} \in \Sigma_s$.
\end{definition}


Concrete examples of structured matrices that both satisfy
Definition~\ref{def:RIP} and admit fast matrix-vector products are discussed in Section~\ref{sec:CompressiveSense} (see also, e.g., \cite[Chapter 12]{FoucartRauhut2013}).

The next theorem can be proven using standard techniques -- we refer the interested reader to the course notes \cite[Section 4.3.3]{Iwennotes} for details.  Ultimately, it will allow us to prove  in Theorem~\ref{theorem:MAM} that matrices $M \in \mathbbm{R}^{m \times N}$ that admit fast matrix-vector multiplications, and/or serve as sublinear-time compressive sensing matrices, can be used to rapidly approximate the top singular values (or eigenvalues) of a huge matrix $A \in \mathbbm{R}^{N \times N}$.  In contrast to related work aimed at approximating eigenvalues of large matrices using unstructured sub-Gaussian matrices (see, e.g., \cite{swartworth2023optimal}), the ability of Theorem~\ref{theorem:MAM} to use arbitrary $\epsilon$-JL maps will also allow efficient (e.g., sublinear-time) sparse approximation of the eigenvectors of $A$.

\begin{theorem}\label{corollary:MAM_A}
    Let $\A\in\mathbbm{C}^{\N\times\N}$ a rank $r$ matrix. Suppose that $\M_1$ and $\M_2$ are $\epsilon$-JL maps of the column spaces of $\A$ and $\A^*$ into $\mathbbm{C}^m$, respectively. Then
    \begin{equation*}
        |\sigma_j(\M_1\A\M_2^*)-\sigma_j(\A)|\leq \epsilon(2+\epsilon)\sigma_j(\A)\quad\forall j\in[\N].
    \end{equation*}
    Here $\sigma_j(\B)$ denotes the $j^{\rm th}$ largest singular value of a given matrix $B$, and $[N] := \{ 1, \dots, N\}$.
\end{theorem}
%
%
Suppose that $\A\in\mathbbm{C}^{\N\times\N}$ is of low rank $\r < \N$. Theorem~\ref{corollary:MAM_A} then implies that we can compute the SVD of $\M_1\A\M_2^* \in \mathbbm{R}^{\mathcal{O}(\r)\times \mathcal{O}(\r)}$ in order to approximate the $r$ non-zero singular values of $A$.  More realistically, however, we should instead consider the case where $A$ is only approximately low rank. 

%
%
%
 Given arbitrary matrix $\A\in\mathbbm{C}^{\N\times\N}$ we can always split it in SVD form as 
\begin{align*}
    \A &= \U\begin{bmatrix}
        \Sigma_\r & 0\\ 0 & \Sigma_{\N-\r}
    \end{bmatrix}
    \V^*\\
     &= \U\begin{bmatrix}
        \Sigma_\r & 0\\ 0 & 0
    \end{bmatrix}
    \V^*
    +\U\begin{bmatrix}
        0 & 0\\ 0 & \Sigma_{\N-\r}
    \end{bmatrix}
    \V^*\\
     &=:\A_\r + \A_{\backslash\r}.
\end{align*}

\begin{theorem}\label{theorem:MAM}
    Let $\A\in\mathbbm{C}^{\N\times\N}$ and choose $r\in[\N]$. Furthermore, suppose that $\M\in\mathbbm{C}^{\m\times\N}$ satisfies the following:
    \begin{enumerate}
        \item $\M$ is an $\epsilon$-JL map of the column space of $\A_\r$ into $\mathbbm{C}^\m$.
        \item $\M$ is an $\epsilon$-JL map of the column space of $\A_\r^*$ into $\mathbbm{C}^\m$.
        \item  M is an $\epsilon$-JL map of the smallest $N-r$ right and $N-r$ left singular vectors of A.
    \end{enumerate}
    Then,
    \begin{align*}
        |\sigma_j(\M\A\M^*) - \sigma_j(\A)|
        \leq \epsilon(2+\epsilon)\sigma_j(\A) + (1+\epsilon) \|A_{\backslash \r}\|_*.
    \end{align*}
    holds $\forall j\in[\r] := \{1, \dots, r\}$.
\end{theorem}
\begin{proof}  Let $j \in [r]$.  Since $\M\A\M^* = \M\A_\r\M^* + \M\A_{\backslash\r}\M^*$ we have 
    \begin{equation*}
        |\sigma_j(\M\A\M^*) - \sigma_j(\M\A_\r\M^*)| \leq \sigma_1(\M\A_{\backslash\r}\M^*)
    \end{equation*}
    by $(c)$ of Theorem 3.3.16 in \cite{horn1991topics}. Furthermore, the right-hand side can be bounded using our third assumption by
    \begin{align*}
    \sigma_1(\M\A_{\backslash r}\M^*)
    &\leq \left\| \sum^N_{j=r+1}\sigma_{j} (A)M\u_j  \v^*_j M^*   \right\|_F \\
        &\leq \sum^N_{j=r+1} \sigma_j (A) \|\left(M\u_j\right)\left(M\v_j\right)^*\|_F \\
        & \leq (1 + \epsilon)\sum^N_{j=r+1} \sigma_j (A) = (1+\epsilon) \|A_{\backslash \r}\|_*.
    \end{align*}
    Hence, by triangle inequality and Theorem~\ref{corollary:MAM_A},
    \begin{align*}
        |\sigma_j(\M\A\M^*) - \sigma_j(\A)|
        &\leq |\sigma_j(\M\A\M^*) - \sigma_j(\M\A_r\M^*)|
        + |\sigma_j(\M\A_r\M^*) - \sigma_j(\A)|\\
        &\leq (1+\epsilon) \|A_{\backslash \r}\|_* + \epsilon(2+\epsilon)\sigma_j(\A).
    \end{align*}
    This completes the proof.
\end{proof}

Restricting our attention now to Hermitian and PSD $A$ with $\lambda_j$, $\lambdatilde_j$, and $\lambdahat_j$ defined as in \eqref{def:SVD_A} -- \eqref{def:lambda}, we can see that under the conditions of Theorem~\ref{theorem:MAM} we have $|\lambda_j - \lambdatilde_j| \leq \epsilon(2+\epsilon) \lambda_j + (1+\epsilon) \|A_{\backslash \r}\|_*$ for all $j \in [r]$.  This result allows us to easily prove the following corollary.

\begin{corollary}\label{coro:Ratioboundtildeorig}
Let $\A\in \mathbbm{C}^{\N\times\N}$, $\r\in[\N]$, and $\epsilon \in (0,1/3)$. Suppose that $\M\in\mathbbm{C}^{\m\times\N}$ satisfies the assumptions of Theorem \ref{theorem:MAM}.  Then, for all $j \in [\r]$,
\begin{equation*}
    \frac{\lambda_j}{\lambdatilde_j} 
    \leq \frac{1}{1-\epsilon(2+\epsilon)} + \frac{1+\epsilon}{1-\epsilon(2+\epsilon)}\frac{ \|A_{\backslash \r}\|_*}{\lambdatilde_j}.
\end{equation*}
\end{corollary}

\begin{proof}
Applying Theorem~\ref{theorem:MAM} we can see that 
\begin{align*}
    \frac{\lambda_j}{\lambdatilde_j} - 1 \leq \epsilon(2+\epsilon) \frac{\lambda_j}{\lambdatilde_j} + (1+\epsilon) \frac{\|A_{\backslash \r}\|_*}{\lambdatilde_j}.
\end{align*}
Rearranging yields
\begin{align*}
    \frac{\lambda_j}{\lambdatilde_j} \left( 1 - \epsilon(2+\epsilon) \right) \leq 1 + (1+\epsilon) \frac{\|A_{\backslash \r}\|_*}{\lambdatilde_j}.
\end{align*}
Dividing through by $1 - \epsilon(2+\epsilon) > 0$ finishes the proof.
\end{proof}

Using Theorem~\ref{theorem:MAM} and Corollary~\ref{coro:Ratioboundtildeorig} we can now also bound $| \lambdatilde_j - \lambdahat_j |$ in terms of $\lambdatilde_j$.  This will be useful in the next section.

\begin{lemma}
\label{lem:distance_lambda_j_lambda_tilde_j}
Let $\A\in \mathbbm{C}^{\N\times\N}$, $\r\in[\N]$, and $\epsilon\in \left(0, \frac{\sqrt{2}-1}{2} \right)$. Suppose that $\M\in\mathbbm{C}^{\m\times\N}$ satisfies the assumptions of Theorem \ref{theorem:MAM}. Define $\lambdahat_j$ as in \eqref{def:lambda}. Then, for all $j \in [\r]$,
\begin{equation*}
    | \lambdatilde_j - \lambdahat_j | \leq \epsilon b_\epsilon \lambdatilde_j + c_{\epsilon,r},
\end{equation*}
where
\begin{align*}
    b_\epsilon := \frac{3+\epsilon}{1-\epsilon(2+\epsilon)} < 6,~~{\rm and}\\  c_{\epsilon,r} := \frac{(1+\epsilon)^2}{1-\epsilon(2+\epsilon)}\|A_{\backslash \r}\|_* < 3 \|A_{\backslash \r}\|_*.
\end{align*}
\end{lemma}
\begin{proof}
Appealing to Theorem~\ref{theorem:MAM} one can see that
\begin{align*}
    \left| \lambdatilde_j - \lambdahat_j \right| &= \left| \lambdatilde_j - \lambda_j\|\M\u_j\|^2_2 \right| \nonumber \\
    &= \left| \lambdatilde_j - \lambda_j + \lambda_j(1 - \|\M\u_j\|^2_2 )\right| \nonumber \\
    &\leq \left| \lambdatilde_j - \lambda_j \right| + \lambda_j \left|1 - \|\M\u_j\|_2^2 \right| \nonumber \\
    &\leq \epsilon(2+\epsilon) \lambda_j + (1+\epsilon) \|A_{\backslash \r}\|_* + \epsilon \lambda_j \nonumber \\
    &= \epsilon(3+\epsilon) \lambda_j + (1+\epsilon) \|A_{\backslash \r}\|_*, 
\end{align*}
where we have used the fact that eigenvectors have $\|\u_j\|_2 = 1$ $\forall j$.  Continuing, we can now further see that 
\begin{align*}
    \left| \lambdatilde_j - \lambdahat_j \right| &\leq \epsilon(3+\epsilon) \lambdatilde_j \left(\frac{\lambda_j}{\lambdatilde_j} \right)+ (1+\epsilon) \|A_{\backslash \r}\|_*\\
    &\leq   \left(\frac{\epsilon(3+\epsilon)}{1-\epsilon(2+\epsilon)}\right) \lambdatilde_j + \frac{(1+\epsilon)^2}{1-\epsilon(2+\epsilon)}\|A_{\backslash \r}\|_*\\
    &= \epsilon b_\epsilon \lambdatilde_j + c_{\epsilon,r}
\end{align*}
by Corollary~\ref{coro:Ratioboundtildeorig}.

It remains to upper bound both $b_\epsilon$ and $c_{\epsilon,r}$.  Since $b_\epsilon$ and $c_{\epsilon,r}$ are both increasing with respect to $\epsilon$ for $\forall \epsilon\in (0, \frac{\sqrt{2}-1}{2})$, they are upper bounded here by
\begin{align*}
    b_\epsilon &\leq \frac{3+\frac{\sqrt{2}-1}{2}}{1-\frac{\sqrt{2}-1}{2}(2+\frac{\sqrt{2}-1}{2})}\\
    &= \frac{12+2\sqrt{2}-2}{4-(\sqrt{2}-1)(4+\sqrt{2}-1)}
    = \frac{10+2\sqrt{2}}{5-2\sqrt{2}}
    < 6, ~{\rm and}\\
    c_{\epsilon,r} &\leq \frac{(1+\frac{\sqrt{2}-1}{2})^2}{1-(\frac{\sqrt{2}-1}{2})(2+\frac{\sqrt{2}-1}{2})} \|\A_{\backslash \r}\|_*
    = \frac{3+2\sqrt{2}}{5-2\sqrt{2}}\|\A_{\backslash \r}\|_*\\
    &<3\|\A_{\backslash \r}\|_*,
\end{align*}
respectively.
\end{proof}

The following corollary of Lemma~\ref{lem:distance_lambda_j_lambda_tilde_j} will be crucial in the next section.

\begin{corollary}\label{coro:spectralgapbounds}
Let $\A\in \mathbbm{C}^{\N\times\N}$, $\r\in[\N]$, and $\epsilon\in \left(0, \frac{\sqrt{2}-1}{2} \right)$. Suppose that $\M\in\mathbbm{C}^{\m\times\N}$ satisfies the assumptions of Theorem \ref{theorem:MAM} and define $b_\epsilon, c_{\epsilon,r}$ as in Lemma~\ref{lem:distance_lambda_j_lambda_tilde_j}.  Then, $\forall j \in [r]$
\begin{equation*}
    \frac{\lambdatilde_j-\lambdahat_j}{\lambdatilde_j} \leq \epsilon b_{\epsilon} + c_{\epsilon,r} \left(\frac{1}{\lambdatilde_j} \right) < 6 \epsilon + 3 \frac{\|\A_{\backslash \r}\|_*}{\lambdatilde_j},
\end{equation*}
and
\begin{align*}
    \frac{\lambdatilde_j}{\lambdahat_j} 
    \geq 1- \epsilon b_{\epsilon} - \frac{c_{\epsilon,r}}{\lambdatilde_j} > 1 - 6 \epsilon - 3 \frac{\|\A_{\backslash \r}\|_*}{\lambdatilde_j}.
\end{align*}
\end{corollary}

\begin{proof}
The first line of inequalities follow immediately from Lemma~\ref{lem:distance_lambda_j_lambda_tilde_j}.  To obtain the second line of inequalities we note from Lemma~\ref{lem:distance_lambda_j_lambda_tilde_j} that 
\begin{align*}
    1 - \frac{\lambdatilde_j}{\lambdahat_j} ~\leq \epsilon b_{\epsilon} \frac{\lambdatilde_j}{\lambdahat_j} + \frac{c_{\epsilon,r}}{\lambdahat_j} ~=~ \epsilon b_{\epsilon} \frac{\lambdatilde_j}{\lambdahat_j} + \frac{c_{\epsilon,r}}{\lambdatilde_j} \frac{\lambdatilde_j}{\lambdahat_j}.
\end{align*}
Rearranging yields
\begin{align*}
    &\frac{\lambdatilde_j}{\lambdahat_j}\left(1 + \epsilon b_{\epsilon} + \frac{c_{\epsilon,r}}{\lambdatilde_j} \right)\geq 1.
\end{align*}
Dividing by $1 + \epsilon b_{\epsilon} + \frac{c_{\epsilon,r}}{\lambdatilde_j}$ and noting that $\frac{1}{1+\left(\epsilon b_{\epsilon} +\frac{c_{\epsilon,r}}{\lambdatilde_j}\right)}\geq 1-\left(\epsilon b_{\epsilon} + \frac{c_{\epsilon,r}}{\lambdatilde_j} \right)$ now finishes the proof. 
\end{proof}

 Having reviewed and utilized standard methods for approximating the singular values of a matrix $\A$ from $\M \A \M^*$ sketches, we will now consider how accurately we can recover compressive measurements of the eigenvectors of $\A$ from such sketches.  In particular, we will continue to assume that $A$ is a Hermitian and PSD matrix below.

\section{Obtaining Compressive Measurements of the Leading Eigenvectors of $\A$ from the Eigenvectors of $\M \A \M^*$.  How Close Are They?}
\label{sec:EigenvectorMeasApprox}

Recall that $\Atilde=\M\A\M^*$ and
\begin{align*}
    \A = \sum_{j=1}^\N \lambda_j\cdot\u_j\u_j^*, \quad
    \Atilde = \sum_{j=1}^\m \lambdatilde_j\cdot\utilde_j\utilde_j^* = \sum_{j=1}^N \lambdahat_j\cdot\uhat_j\uhat_j^*
\end{align*}
where
\begin{equation*}
    \lambdahat_j = \lambda_j\|\M\u_j\|^2_2,~~
    \uhat_j = \frac{\M\u_j}{\|M\u_j\|_2}.
\end{equation*}
We now study the relationship between $\utilde_j$ and $\uhat_j$. For notational simplicity, we define the \emph{relative gaps} between eigenvalues of $\Atilde$ as
\begin{equation}\label{eq:gap}
    g_j := \frac{\lambdatilde_j - \lambdatilde_{j+1}}{\lambdatilde_j}, \quad \forall j \in \{1,2,\cdots,m\}.
\end{equation}

Note that these relative gaps are directly computable from the sketch $\Atilde$, and that they should be close to the true relative gaps between the top eigenvectors of all approximately low-rank $\A$ by Theorem~\ref{theorem:MAM}.  The following theorem is proven in Section~\ref{PROOFOFthm:angle_lower_bound}.

\begin{theorem}
\label{thm:angle_lower_bound}
Let $\A\in \mathbbm{C}^{\N\times\N}$ be Hermitian and PSD, and $\epsilon\in \left(0, \frac{\sqrt{2}-1}{2} \right)$. Fix $\r\in[\N]$. Suppose $\lambdatilde_j$ are all distinct and that $\M \in\mathbbm{C}^{\m\times\N}$ satisfies the assumptions in Theorem \ref{theorem:MAM}. Define $\utilde_j$ and $\uhat_j$ as in \eqref{def:SVD_A_Atilde} and \eqref{def:lambda}. Then the following holds for every $j \in [r]$:
\begin{equation*}
    |\langle \utilde_j,\uhat_j \rangle|^2 
    \geq 1-6\epsilon \left(1+\frac{\kappa_j}{g_j} \right) - \frac{3\|\A_{\backslash \r}\|_*}{\lambdatilde_j}\left(1+\frac{\nu_j}{ g_j} \right),
\end{equation*}
where $g_j$ are the gaps defined as in \eqref{eq:gap}, and both $\nu_j$ and $\kappa_j$ are recursively defined via 
\begin{align*}
    \kappa_j&=1 + \frac{1}{\lambdatilde_j}\sum_{\ell=1}^{j-1}  \frac{\lambdatilde_\ell \kappa_\ell}{g_\ell} \quad {\rm and} \quad
    \nu_j = 1+\sum_{\ell=1}^{j-1} \frac{\nu_\ell}{g_\ell}.
\end{align*}
\end{theorem}

Note that if $\|\A_{\backslash \r}\|_*$ and $\epsilon$ are both small then $\Atilde$'s $j^{\rm th}$ eigenvector $\utilde_j \approx \frac{\M\u_j}{\|M\u_j\|_2} \approx (1 \pm \mathcal{O}(\epsilon)) M\u_j$.  That is, $\utilde_j$ will provide an accurate compressive measurement of its associated eigenvector of $\A$.  The following corollary of Theorem~\ref{thm:angle_lower_bound} makes this precise with respect to Euclidean distance.

\begin{corollary} \label{Coro:approxerrorinEigenvectors}
Let $\A\in \mathbbm{C}^{\N\times\N}$ be Hermitian and PSD, and $\epsilon\in \left(0, \frac{\sqrt{2}-1}{2} \right)$. Fix $\r\in[\N]$. Suppose $\lambdatilde_j$ are all distinct and that $\M \in\mathbbm{C}^{\m\times\N}$ satisfies the assumptions in Theorem \ref{theorem:MAM}. Define $\utilde_j$ as in \eqref{def:SVD_A_Atilde}, $\u_j$ as in \eqref{def:SVD_A}, $g_j$ as in \eqref{eq:gap}, and both $\nu_j$ and $\kappa_j$ as in Theorem~\ref{thm:angle_lower_bound}.  Then,
\begin{align*}
    \min_{\phi \in [0,2\pi)} \| \mathbbm{e}^{\mathbbm{i} \phi} \utilde_j - M \u_j \|_2 < \sqrt{12\epsilon \left(\frac{5}{4}+\frac{\kappa_j}{g_j} \right) + \frac{6\|\A_{\backslash \r}\|_*}{\lambdatilde_j}\left(1+\frac{\nu_j}{ g_j} \right)}
\end{align*}
holds for all $j \in [r]$.
\end{corollary}

\begin{proof}
Fix $j \in [r]$ and let $\phi' \in [0,2\pi)$ be such that $\langle \mathbbm{e}^{\mathbbm{i} \phi'} \utilde_j, M \u_j \rangle = | \langle \utilde_j, M \u_j \rangle|  = \| M \u_j \|_2 | \langle \utilde_j, \uhat_j \rangle|$.  Then,
\begin{align*}
    \| \mathbbm{e}^{\mathbbm{i} \phi'} \utilde_j - M \u_j \|_2^2 &= 1 + \|M \u_j\|_2^2 - 2 \langle \mathbbm{e}^{\mathbbm{i} \phi'} \utilde_j, M \u_j \rangle \\
    &= 1 + \|M \u_j\|_2^2 - 2 \| M \u_j \|_2 | \langle \utilde_j, \uhat_j \rangle|\\
    &\leq 2 - 2 \| M \u_j \|_2 | \langle \utilde_j, \uhat_j \rangle| + \epsilon.
\end{align*}
Using this last line we can now in fact see that $\| \mathbbm{e}^{\mathbbm{i} \phi'} \utilde_j - M \u_j \|_2^2$ is upper bounded by $\left(2 - 2 \| M \u_j \|_2 | \langle \utilde_j, \uhat_j \rangle| \right) \left( 1 + \| M \u_j \|_2 | \langle \utilde_j, \uhat_j \rangle| \right)+ \epsilon.$
Simplifying this last expression we finally learn that 
\begin{align*}
\| \mathbbm{e}^{\mathbbm{i} \phi'} \utilde_j - M \u_j \|_2^2 &\leq 2 - 2 \| M \u_j \|_2^2 | \langle \utilde_j, \uhat_j \rangle|^2 + \epsilon\\
&\leq 2 - 2 (1 - \epsilon) \left( 1-6\epsilon \left(1+\frac{\kappa_j}{g_j} \right) - \frac{3\|\A_{\backslash \r}\|_*}{\lambdatilde_j}\left(1+\frac{\nu_j}{ g_j} \right)\right)+ \epsilon\\
&= 2(1-\epsilon)\left( 6\epsilon \left(1+\frac{\kappa_j}{g_j} \right) + \frac{3\|\A_{\backslash \r}\|_*}{\lambdatilde_j}\left(1+\frac{\nu_j}{ g_j} \right)\right) + 3 \epsilon\\
&< 12\epsilon \left(\frac{5}{4}+\frac{\kappa_j}{g_j} \right) + \frac{6\|\A_{\backslash \r}\|_*}{\lambdatilde_j}\left(1+\frac{\nu_j}{ g_j} \right).
\end{align*}
The result now follows after taking a square root. \end{proof}

In order to get a better sense for how well Corollary~\ref{Coro:approxerrorinEigenvectors} is guaranteed to work in practice, however, it'd be more convenient to have the quantities $\kappa_j, \nu_j$ from Theorem~\ref{thm:angle_lower_bound} stated in terms of the spectral properties of $\A$ instead of $\Atilde$.  We will now do this next for a particular class of matrices $\A$ that have exponentially decaying eigenvalues.  The following theorem is proven in Section~\ref{PROOFOFMainTHM:ExpDecayingSingVlaues}.

\begin{theorem} \label{MainTHM:ExpDecayingSingVlaues}
Let $q \in (0,1/3)$, $c \in [1,\infty)$, $\epsilon < \min \left\{ \frac{1}{20}, \frac{1}{4}\left(\frac{1-3q}{1+q}\right) \right\}$, and $\ell, r \in [N]$ with $2 \leq \ell \leq r$. Suppose that $\A\in \mathbbm{C}^{\N\times\N}$ is Hermitian and PSD with eigenvalues satisfying
\begin{enumerate}
    \item $\lambda_j = c q^j$ for all $j \in [\ell] \subseteq [r]$, and
    \item $\|\A_{\backslash \r}\|_* \leq \epsilon \lambda_\ell$.
\end{enumerate}
In addition, suppose that $\M \in\mathbbm{C}^{\m\times\N}$ satisfies the assumptions in Theorem \ref{theorem:MAM}, and that $\utilde_j$ and $\u_j$ are defined as in \eqref{def:SVD_A_Atilde} and \eqref{def:SVD_A}.  Then, \begin{align*}
    \min_{\phi \in [0,2\pi)} \| \mathbbm{e}^{\mathbbm{i} \phi} \utilde_j -& M \u_j \|_2 < 7 \sqrt{\epsilon} \cdot q^{1-j}
\end{align*}
holds for all $j \in [\ell]$.
\end{theorem}

\subsection{Proof of Theorem~\ref{thm:angle_lower_bound}} \label{PROOFOFthm:angle_lower_bound}

Since $\A$ is PSD, $\M\A\M^*$ is also PSD. Recall that an SVD of $\Atilde$ is given by $\Atilde=\Utilde\Sigmatilde\Utilde^*$; in analogy we set $\Uhat\in\mathbbm{C}^{\m\times\N}$ to be the matrices whose $j$-th columns are $\sqrt{\lambdahat_j}\uhat_j$, so that $\Atilde = \Uhat\Uhat^*$. Note that $\Utilde$ is orthogonal but $\Uhat$ is not. By equating the two expressions of $\Atilde$ we know that
\begin{equation*}
    \Utilde\Sigmatilde\Utilde^* = \Uhat\Uhat^*.
\end{equation*}
Together with orthogonality of $\Utilde$ we obtain
\begin{equation*}
    \Sigmatilde = \Utilde^*\Uhat\Uhat^*\Utilde =: \R\R^*.
\end{equation*}
Let $\R=\Q\D\W^*$ be an SVD where $\Q\in\mathbbm{C}^{\m\times\m}$ and $\W\in\mathbbm{C}^{\N\times\N}$ are orthogonal, and $\D\in\mathbbm{R}^{\m\times\N}$ is diagonal. Then $\Q$ and $\D$ must satisfy the relation $\Sigmatilde = \Q\D^2\Q^*$. From the last relation we deduce that, since the diagonal entries of $\Sigmatilde$ are distinct,
\begin{align*}
    Q = I \text{ and } D_{ij} = \begin{cases}
      \sqrt{\Sigmatilde_{ij}}\quad\text{if}\quad 1 \le i=j \le m\\
      0 \qquad\quad \text{else}
    \end{cases}.
\end{align*}
By equating the definition of $\R$ ($=\Utilde^*\Uhat$) and its SVD ($\D\W^*$), we get that
\begin{align*}
    \sqrt{\lambdahat_j}\langle \utilde_j, \uhat_j \rangle
    = (\Utilde^*\Uhat)_{j,j}
    = \R_{j,j}
    = (\D\W^*)_{j,j}
    = \sqrt{\lambdatilde_j}\W_{j,j}^*
\end{align*}
where we recall that $\lambdatilde_j:=\sigma_j(\Atilde)$.  Thus
\begin{equation}
    |\langle \utilde_j, \uhat_j \rangle|^2 = \frac{\lambdatilde_j}{\lambdahat_j}|\W_{j,j}|^2. \label{equ:InnerProdFormula}
\end{equation}
Since the fraction $\frac{\lambdatilde_j}{\lambdahat_j}$ has already been bounded in Corollary~\ref{coro:spectralgapbounds}, we will now focus on deriving a lower bound for $|W_{j,j}|^2$. 

To begin we note that 
\begin{align*}
\Uhat & = \Utilde\R \qquad \text{ where }\quad \Utilde \text{ is orthogonal and }\quad \R=\D\W^*. 
\end{align*}
In addition we have
\begin{align*}
    (D^*D)_{ij} = \begin{cases}
        \Sigmatilde_{ii} &\quad \text{if}\quad 1 \le i = j \le m \\
        0 ,& \quad \text{else}
                \end{cases}.
\end{align*}
Hence, we have that
\begin{align}
    \lambdahat_j
    &= (\Uhat^*\Uhat)_{j,j}
    = (\R^*\Utilde^*\Utilde\R)_{j,j}
    = (\W\D^*\D\W^*)_{j,j} \nonumber \\
    &= \sum_{k=1}^\m \lambdatilde_k |W_{j,k}|^2. \label{equ:forWjjbound}
\end{align}
We will now proceed to use $\lambdahat_j$ expressed in terms of the entries of $\W$ to bound $W_{j,j}$ in terms of eigenvalue data of $\Atilde$ and $\A$.

Rewriting \eqref{equ:forWjjbound} we have
\begin{equation*}
    \lambdahat_j = \sum_{k=1}^{j-1} \lambdatilde_k |W_{j,k}|^2 + \lambdatilde_j |W_{j,j}|^2 + \sum_{k=j+1}^\m \lambdatilde_k |W_{j,k}|^2,
\end{equation*}
which can be upper bounded by
\begin{align*}
    \lambdahat_j
    &\leq \sum_{k=1}^{j-1} \lambdatilde_k|W_{j,k}|^2 + \lambdatilde_j |W_{j,j}|^2 + \lambdatilde_{j+1}\left(1 - \sum_{k=1}^{j} |W_{j,k}|^2 \right)\\
    &\leq \sum_{k=1}^{j-1} \lambdatilde_k(1-|W_{k,k}|^2)  + \lambdatilde_j |W_{j,j}|^2 + \lambdatilde_{j+1}(1 - |W_{j,j}|^2)
\end{align*}
due to the orthonormality of $W$.
Furthermore, since
\begin{align*}
    \lambdatilde_j |W_{j,j}|^2 &+ \lambdatilde_{j+1}(1 - |W_{j,j}|^2)\\
    &= \lambdatilde_j +  \lambdatilde_j(|W_{j,j}|^2-1) + \lambdatilde_{j+1}(1 - |W_{j,j}|^2)\\
    &= \lambdatilde_j - (\lambdatilde_j-\lambdatilde_{j+1})(1 - |W_{j,j}|^2)\\
    &= \lambdatilde_j - \lambdatilde_jg_j(1 - |W_{j,j}|^2)
\end{align*}
we have that
\begin{equation*}
    \lambdahat_j
    \leq \sum_{k=1}^{j-1} \lambdatilde_k(1-|W_{k,k}|^2)  + \lambdatilde_j - \lambdatilde_jg_j(1 - |W_{j,j}|^2),
\end{equation*}
or equivalently,
\begin{equation}\label{eq:lambt_upper_bound_2}
    g_j(1 - |W_{j,j}|^2)
    \leq \frac{1}{\lambdatilde_j}\sum_{k=1}^{j-1}\lambdatilde_k (1-|W_{k,k}|^2)  + \frac{\lambdatilde_j-\lambdahat_j}{\lambdatilde_j}.
\end{equation}
We will now show by induction that for all $j\in[r]$,
\begin{align}
    g_j(1 - |W_{j,j}|^2) \leq \epsilon b_\epsilon \kappa_j + c_{\epsilon,r} \left(\frac{\nu_j}{\lambdatilde_j}\right) \label{eq:inductive_hypothesis}
\end{align}
where $b_\epsilon$ and $c_{\epsilon,r}$ are as in Lemma~\ref{lem:distance_lambda_j_lambda_tilde_j}.

For $j=1$ Corollary~\ref{coro:spectralgapbounds} together with \eqref{eq:lambt_upper_bound_2} immediately implies that  
\begin{align*}
    g_1(1 - |W_{1,1}|^2) &\leq \frac{\lambdatilde_1 - \lambdahat_1}{\lambdatilde_1} \leq \epsilon b_{\epsilon} + c_{\epsilon,r} \left(\frac{1}{\lambdatilde_1} \right)\\
    &\leq \epsilon b_\epsilon \kappa_1 + c_{\epsilon,r} \left(\frac{\nu_1}{\lambdatilde_1} \right).
\end{align*}
Now suppose that \eqref{eq:inductive_hypothesis} holds for all $\ell \in[j-1]$. That is,
\begin{align*}
   (1-|W_{\ell,\ell}|^2) \leq \epsilon b_\epsilon \frac{\kappa_\ell}{g_\ell} + c_{\epsilon,r} \frac{\nu_\ell}{\lambdatilde_\ell g_\ell} \quad \forall \ell \in [j-1].
 \end{align*}
Then by \eqref{eq:lambt_upper_bound_2} and Corollary~\ref{coro:spectralgapbounds} we obtain
\begin{align*}
    g_j(1 &- |W_{j,j}|^2) \leq \frac{1}{\lambdatilde_j}\sum_{\ell=1}^{j-1} \lambdatilde_\ell(1-|W_{\ell,\ell}|^2)  + \frac{\lambdatilde_j-\lambdahat_j}{\lambdatilde_j}\\
    &\leq \frac{\epsilon b_\epsilon}{\lambdatilde_j} \left(\sum_{\ell=1}^{j-1}  \frac{\lambdatilde_\ell \kappa_\ell}{g_\ell} \right)+ \frac{c_{\epsilon,r}}{\lambdatilde_j} \left(\sum_{\ell=1}^{j-1} \frac{\lambdatilde_\ell \nu_\ell}{\lambdatilde_\ell g_\ell}\right) + \frac{\lambdatilde_j - \lambdahat_j}{\lambdatilde_j}\\
    &\leq \frac{\epsilon b_\epsilon}{\lambdatilde_j}\left(\sum_{\ell=1}^{j-1}  \frac{\lambdatilde_\ell \kappa_\ell}{g_\ell} \right) + \frac{c_{\epsilon,r}}{\lambdatilde_j} \left(\sum_{\ell=1}^{j-1} \frac{\nu_\ell}{g_\ell} \right) + \epsilon b_\epsilon + \frac{c_{\epsilon,r}}{\lambdatilde_j}\\
    &\leq \epsilon b_\epsilon \left(1 + \frac{1}{\lambdatilde_j}\sum_{\ell=1}^{j-1}  \frac{\lambdatilde_\ell \kappa_\ell}{g_\ell}\right) + \frac{c_{\epsilon,r}}{\lambdatilde_j}\left(1+\sum_{\ell=1}^{j-1} \frac{\nu_\ell}{g_\ell}\right)\\
    &= \epsilon b_\epsilon \kappa_j + c_{\epsilon,r}\frac{\nu_j}{\lambdatilde_j}.
\end{align*}
Hence, \eqref{eq:inductive_hypothesis} holds for all $j \in [r]$.

Rearranging \eqref{eq:inductive_hypothesis} we can now see that
\begin{equation*}
     |W_{j,j}|^2 \geq 1-\epsilon b_\epsilon \frac{\kappa_j}{g_j} - c_{\epsilon,r}\frac{\nu_j}{\lambdatilde_j g_j}
\end{equation*}
holds for all $j \in [r]$.  Combining this with \eqref{equ:InnerProdFormula} and Corollary~\ref{coro:spectralgapbounds} we obtain the lower bound 
\begin{align*}
     |\langle \utilde_j, \uhat_j \rangle|^2 &= \frac{\lambdatilde_j}{\lambdahat_j}|\W_{j,j}|^2\\
     &\geq \left( 1- \epsilon b_{\epsilon} - \frac{c_{\epsilon,r}}{\lambdatilde_j} \right)\left(1-\epsilon b_\epsilon \frac{\kappa_j}{g_j} - c_{\epsilon,r}\frac{\nu_j}{\lambdatilde_j g_j}\right)\\
     &\geq 1 -\epsilon b_\epsilon \left(1+\frac{\kappa_j}{g_j} \right) - \frac{c_{\epsilon,r}}{\lambdatilde_j}\left(1+\frac{\nu_j}{ g_j} \right)
\end{align*}
where in the last line we have used that $(1-a_1)(1-a_2) \geq 1 - a_1 - a_1$ when $a_1,a_2\geq 0$. 
Substituting the upper bounds for $b_\epsilon$ and $c_{\epsilon,r}$ from Lemma~\ref{lem:distance_lambda_j_lambda_tilde_j} now finishes the proof.

\subsection{Proof of Theorem~\ref{MainTHM:ExpDecayingSingVlaues}} \label{PROOFOFMainTHM:ExpDecayingSingVlaues}

Since, e.g., $\epsilon < \frac{1}{2}$ and $\|\A_{\backslash \r}\|_* \leq \epsilon \lambda_\ell$, Theorem~\ref{theorem:MAM} implies that $|\lambdatilde_j - \lambda_j|
        \leq 4 \epsilon \lambda_j$
holds for all $j \in [\ell]$.  In addition, $\epsilon < \frac{1}{4}\left(\frac{1-3q}{1+q}\right) \iff 2q + 8 \epsilon q < 1 - 4 \epsilon - q + 4 \epsilon q$ implies that 
\begin{equation} \label{equ:importantprop}
(1 + 4 \epsilon)q < 2(1 + 4 \epsilon)q < (1 - 4 \epsilon) (1-q) < (1 - 4 \epsilon).
\end{equation}
Hence, we can further see that the order of the top $\ell$ eigenvalues of $\A$ is preserved in $\Atilde$.  That is,
$$\lambdatilde_j \leq (1+4\epsilon) c q^j < (1-4\epsilon) c q^{j-1} \leq \lambdatilde_{j-1}$$
holds for all $j \in [\ell] \setminus \{1\}$.

Considering the relative eigengaps of $\Atilde$ we can also see that since $\epsilon < \frac{1}{20}$ and $q < 1/3$ we will have
$$1 \geq g_j = 1 - \frac{\lambdatilde_{j+1}}{\lambdatilde_j} \geq 1 - \frac{(1+4 \epsilon)q}{(1-4 \epsilon)} > 1 - \frac{3}{2}\left(\frac{1}{3}\right) = \frac{1}{2}$$
for all $j \in [\ell]$.  As a consequence we may bound the $\kappa_j$ $\forall j \in [\ell]$ by
\begin{align*}
 \kappa_j &\leq 1 + \frac{q^{-j}}{c(1-4 \epsilon)} \sum_{l = 1}^{j-1} 2cq^{l}(1+4 \epsilon) \kappa_l\\ &= 1 + \frac{2(1+4 \epsilon)q^{-j}}{1-4 \epsilon} \left( \sum_{l = 1}^{j-1} q^{l} \kappa_l \right)\\
 &< 1 + q^{-j}(q^{-1}-1) \sum_{l = 1}^{j-1} q^{l} \kappa_l,
\end{align*}
where we have used \eqref{equ:importantprop} in the last step.
A short induction argument now shows that $\kappa_j \leq q^{-2j+2}$ for all $j \in [\ell]$.  First, $\kappa_1 = 1 = q^{0}$ always holds for $j = 1$. Next, if $\kappa_l \leq q^{-2l+2}$ for all $l \in [j-1]$ with $j \geq 2$, then
\begin{align*}
 \kappa_j &< 1 + q^{-j}(q^{-1}-1) \sum_{l = 1}^{j-1} q^{-l+2}\\
 &<  q^{-j+2}(q^{-1}-1) \left(1 + \sum_{l = 1}^{j-1} q^{-l} \right)\\
 &= q^{-j+2}(q^{-1}-1) \left( \frac{q^{-j} - 1}{q^{-1}-1} \right) < q^{-2j+2}.
\end{align*}

Focusing now on $\nu_j$, another short induction argument shows that $\nu_j \leq 3^{j-1}$ holds for all $j \in [\ell]$.  First, $\nu_1 = 1$.  Now assume that $\nu_l \leq 3^{l-1}$ for all $l \in [j-1]$.  Then, 
\begin{align*}
    \nu_j &< 1 + 2 \sum_{\ell=1}^{j-1} \nu_\ell \leq 1 + 2 \sum_{\ell=1}^{j-1} 3^{l-1}\\ &=1 + 2 \left( \frac{3^{j-1} - 1}{3 - 1} \right)= 3^{j-1}.
\end{align*}

The result now follows by substituting our bounds on $g_j$, $\kappa_j$, $\nu_j$, and $\|\A_{\backslash \r}\|_*$ into Corollary~\ref{Coro:approxerrorinEigenvectors}.  Doing so we learn that 
\begin{align*}
    \min_{\phi \in [0,2\pi)} \| \mathbbm{e}^{\mathbbm{i} \phi}& \utilde_j - M \u_j \|_2 \\ \leq&\sqrt{12\epsilon \left(\frac{5}{4}+\frac{\kappa_j}{g_j} \right) + \frac{6\|\A_{\backslash \r}\|_*}{\lambdatilde_j}\left(1+\frac{\nu_j}{ g_j} \right)}\\
    \leq&\sqrt{12\epsilon \left(\frac{5}{4}+2 q^{-2j+2} \right) + \frac{6 \epsilon \lambda_\ell \left(1+2 \cdot 3^{j-1} \right)}{(1-4 \epsilon) \lambda_j}}\\
    \leq&\sqrt{39\epsilon \cdot q^{-2j+2} + \frac{15}{2} \epsilon \cdot q^{\ell - j} 3^{j}}\\
    \leq&\sqrt{ 46.5 \epsilon \cdot q^{-2j+2} } ~<~ 7 \sqrt{\epsilon} \cdot q^{-j+1} 
\end{align*}
holds for all $j \in [r]$.  In the second-from-last inequality we used that $\epsilon < \frac{1}{20}$ and that $1+2 \cdot 3^{j-1} \leq 3^j$ for all $j \geq 1$.

\section{Applicable Compressive Sensing Algorithms for Rapidly Approximating Eigenvectors}
\label{sec:CompressiveSense}

In this section we consider two possible compressive sensing strategies for recovering sparse approximations of $A$'s top eigenvectors $\u_j$ \eqref{def:SVD_A} from the top eigenvectors $\utilde_j$ of $MAM^*$ \eqref{def:SVD_A_Atilde}.  Before beginning, however, we emphasize here that many other options are also possible.  These are simply two particular examples.

Given that many classical JL-embedding and compressive sensing results are formulated for $\mathbbm{R}^N$ instead of $\mathbbm{C}^N$, we will often (implicitly) use the following lemma below.  It will ultimately allow us to consider more general Hermitian (as opposed to simply symmetric) matrices $A$.\footnote{Note that symmetric PSD matrices are a special case of Hermitian PSD matrices.} 

\begin{lemma}
\label{lem:RealRIPmatricesinCN}
    Let $\mathcal{S} \subset \mathbbm{C}^N$.  If $M \in \mathbbm{R}^{m \times N}$ is an $\epsilon$-JL map of $\tilde{\mathcal{S}} := \left\{ \mathbbm{Re}(\u) ~|~ \u \in \mathcal{S} \right\} \cup \left\{ \mathbbm{Im}(\u) ~|~ \u \in \mathcal{S} \right\} \subset \mathbbm{R}^N$ into $\mathbbm{R}^m$, then $M$ is also an $\epsilon$-JL map of $\mathcal{S}$ into $\mathbbm{C}^m$.
\end{lemma}

\begin{proof}
    Let $\u \in \mathcal{S}$ and suppose that $M \in \mathbbm{R}^{m \times N}$ is an $\epsilon$-JL map of $\tilde{\mathcal{S}}$.  Then, 
    \begin{align*}
       (1-\epsilon)\| \u \|_2^2
       &= (1-\epsilon)\left( \left\| \mathbbm{Re}(\u) \right\|_2^2 +  \left\| \mathbbm{Im}(\u)  \right\|_2^2 \right)
       \leq \left\|M \mathbbm{Re}(\u) \right\|_2^2 +  \left\| M \mathbbm{Im}(\u)  \right\|_2^2\\
       &=\left\|\mathbbm{Re}(M\u) \right\|_2^2 +  \left\| \mathbbm{Im}(M\u)  \right\|_2^2 = \| M \u \|_2^2\\
       &\leq (1+\epsilon)\left( \left\| \mathbbm{Re}(\u) \right\|_2^2 +  \left\| \mathbbm{Im}(\u)  \right\|_2^2 \right)
       = (1+\epsilon)\| \u \|_2^2.
    \end{align*}
\end{proof}

Looking at Lemma~\ref{lem:RealRIPmatricesinCN} one can see that if $\mathcal{S} \subset \mathbbm{C}^N$ is finite, then $|\tilde{\mathcal{S}}| \leq 2 |\mathcal{S}|$.  In addition, if $\mathcal{S}$ is an $r$-dimensional subspace of $\mathbbm{C}^N$, then $\tilde{\mathcal{S}}$ is contained in a $2r$-dimensional subspace of $\mathbbm{R}^N$.  Finally, if $\mathcal{S} = \Sigma_{s} = \{ {\bf x} ~|~ \| {\bf x} \|_0 \leq s \} \subset \mathbbm{C}^N$, then $\tilde{\mathcal{S}} = \Sigma_{s} \cap \mathbbm{R}^N$. All three of these facts will be used liberally below.

\subsection{A Linear-Time Compressive Sensing Algorithm}

We will first present a linear-in-$N$-time strategy for approximating the eigenvectors of $A$ based on the CoSaMP algorithm \cite{needell2009cosamp}.  The performance of this method is summarized in the next Theorem (see, e.g., \cite[Theorem A]{needell2009cosamp}).  

\begin{theorem} \label{Thm:CoSamp} Let $s \in [N]$, $\eta \in (0,1)$, and 
    suppose that $M \in \mathbbm{R}^{m \times N}$ has both $(i)$ the RIP of order $\left(4s,0.1 \right)$ (see Definition~\ref{def:RIP}), and $(ii)$ an associated $\mathcal{O}(N \log N)$-time matrix-vector multiplication algorithm for both $M$ and $M^*$.  Then, there exists a compressive sensing algorithm $g: \mathbbm{C}^m \rightarrow \mathbbm{C}^N$ such that 
\begin{equation*}
        \| \u - g \left(M \u + {\bf e} \right) \|_2 \leq c \cdot \max \left\{ \eta, \frac{1}{\sqrt{s}} \| {\u} - \u_s\|_1 + \| {\bf e} \|_2 \right\}
\end{equation*}
    holds for all $\u \in \mathbbm{C}^N$ and ${\bf e} \in \mathbbm{C}^m$,
    where $c \in \mathbbm{R}^+$ is an absolute/universal constant.  Furthermore, $g$ can always be evaluated in $\mathcal{O}(N \log N \cdot \log(\| \u\|_2/\eta))$-time.
\end{theorem}

Let $H \in \mathbbm{R}^{N \times N}$ be a symmetric unitary Hadamard matrix with an $\mathcal{O}(N \log N)$-time matrix-vector multiply.\footnote{The \emph{Hadamard matrix of order $N = 2^k$} is a matrix with real entries satisfying $HH^\top = I$.  Hadamard matrices $H \in \mathbbm{R}^{N \times N}$ can be multiplied against vectors in 
$\mathcal{O}(N \log N)$-time via the \emph{Walsh--Hadamard transform}, a real-valued
analog of the FFT (see, e.g., \cite{ahmed2012orthogonal,hedayat1978hadamard,shanks1969computation}).  We assume here that $N$ is a power of $2$ for simplicity.  If not, one can simply (implicitly) pad all vectors involved with $0$s until they are of length $2^{\lceil \log_2 N \rceil}$.}  Let $R \in \{ 0,1\}^{m \times N}$ be a random matrix created by independently selecting $m$ rows of the $N \times N$ identity matrix $I$ uniformly at random with replacement.  Finally, let $D \in \{ -1,0,1\}^{N \times N}$ be a random diagonal matrix with independent and identically distributed (i.i.d.) Rademacher random variables (i.e., $\pm 1$ with probability $1/2$) on its diagonal, and set
\begin{equation} \label{equ:DefCosampM}
    M := \sqrt{\frac{N}{m}}RHD \in \mathbbm{R}^{m \times N}.
\end{equation}   

Note that $M$ in \eqref{equ:DefCosampM} will have an associated $\mathcal{O}(N \log N)$-time matrix-vector multiplication algorithm for both $M$ and $M^*$ by construction (i.e., via the Walsh–Hadamard transform).  In addition, when $M$ has a sufficiently large number of rows it will also both $(i)$ have the RIP of order $\left(4s,0.1 \right)$, and $(ii)$ satisfy the assumptions in Theorem \ref{theorem:MAM} with respect to an arbitrary $A \in \mathbbm{C}^{N \times N}$ with high probability.  The next lemma guarantees that $M \in \mathbbm{R}^{m \times N}$ will simultaneously satisfy all of these useful properties whenever $m$ is sufficiently large.

\begin{lemma} \label{lem:MergedPropsCosamp}
Let $\A\in\mathbbm{C}^{\N\times\N}$, $s, r \in [N]$, and $\epsilon, p \in (0,1)$.  If 
$$N ~\geq~ m ~\geq~ c_1 \max\{s,r/\epsilon^2\} \log^4\left(c_2 N/p\epsilon^2 \right)$$ 
then the matrix $\M\in\mathbbm{R}^{\m\times\N}$ in \eqref{equ:DefCosampM} will simultaneously satisfy all of the following properties with probability at least $1-p$:
    \begin{enumerate}
    \item $\M$ will have the RIP of order $\left(4s,0.1 \right)$,
        \item $\M$ will be an $\epsilon$-JL map of the column space of $\A_\r$ into $\mathbbm{C}^\m$,
        \item $\M$ will be an $\epsilon$-JL map of the column space of $\A_\r^*$ into $\mathbbm{C}^\m$, and
        \item  $\M$ will be an $\epsilon$-JL map of the smallest $N-r$ right and $N-r$ left singular vectors of A.
    \end{enumerate}
Here, $c_1, c_2 \in \mathbbm{R}^+$ are absolute/universal constants.
\end{lemma}

\begin{proof}
    Noting that $D \Sigma_s = \Sigma_s$ for all $s \in [N]$, one can apply \cite[Theorem 12.32]{FoucartRauhut2013} in order to see that $M$ will satisfy property 1 above with probability at least $1-(p/4)$ since $m \geq c_1' s \cdot \left(  \log^2(4s)\log^2(9N)+ \log(4/p)\right)$.  Considering properties 2 and 3 above, we can see that they will also each be satisfied with probability at least $1-(p/4)$ by \cite[Theorem 2.6]{iwen2024on} since $m \geq \frac{c''_1}{\epsilon^2} r \cdot \log^2 \left(\frac{c_2' r \log(8/p)}{\epsilon^2} \right)\log(8/p) \log(8 \mathbbm{e} N/p)$, where we have also used that the Gaussian width \cite{vershynin2018high} of a $2r$-dimensional subspace of $\mathbbm{R}^N$ is bounded above by $\sqrt{2r}$.  Finally, considering property 4 above, we can see that it will also be satisfied with probability at least $1-(p/4)$ by \cite[Corollary 2.5]{iwen2024on} since $m \geq \frac{c_1'''}{\epsilon^2} \log\left( \frac{c_2'' N}{p}  \right) \cdot \left(\log^2\left( \frac{\log(c_2''' N/p)}{\epsilon}  \right) \log(\mathbbm{e}N) + \log(8 \mathbbm{e}/p) \right)$. Setting $c_1 = \max\{2c_1', c_1'', 2c_1'''\}$, $c_2 = \max\{8c_2', c_2'', c_2''',8\mathbbm{e}\}$, and applying the union bound now finishes the proof.
\end{proof}

Our main result of this subsection now follows immediately from Theorem~\ref{Thm:CoSamp} and Lemma~\ref{lem:MergedPropsCosamp}.

\begin{theorem} \label{Thm:MainCosampMatrixSetup}
Let $\A\in\mathbbm{C}^{\N\times\N}$, $s, r \in [N]$, and $\epsilon, p, \eta \in (0,1)$.  If 
$$N ~\geq~ m ~\geq~ c_1 \max\{s,r/\epsilon^2\} \log^4\left(c_2 N/p\epsilon^2 \right)$$ 
then with probability at least $1-p$ the matrix $\M\in\mathbbm{R}^{\m\times\N}$ in \eqref{equ:DefCosampM} will both $(i)$ satisfy all 3 assumptions of Theorem \ref{theorem:MAM}, and $(ii)$ satisfy
\begin{equation*}
        \| \u - g \left(M \u + {\bf e} \right) \|_2 \leq c \cdot \max \left\{ \eta, \frac{1}{\sqrt{s}} \| {\u} - \u_s\|_1 + \| {\bf e} \|_2 \right\}
\end{equation*}
    for all $\u \in \mathbbm{C}^N$ and ${\bf e} \in \mathbbm{C}^m$.  Here $g: \mathbbm{C}^m \rightarrow \mathbbm{C}^N$ is as in Theorem~\ref{Thm:CoSamp} so that evaluations of $g$ can be computed in $\mathcal{O}(N \log N \cdot \log(\| \u\|_2/\eta))$-time, and $c, c_1, c_2 \in \mathbbm{R}^+$ are absolute/universal constants.
\end{theorem}

We will now turn our attention to an algorithm for approximation the eigenvectors of $A \in \mathbbm{C}^{N \times N}$ whose runtime scales sublinearly in $N$.

\subsection{A Sublinear-Time Compressive Sensing Algorithm}

In this subsection we will utilize slightly generalized versions of the compressive sensing algorithms in \cite{BaileyIwenSpencer2012,iwen2014} which are themselves generalized variants of the compressive sensing strategy first proposed by Cormode and  Muthukrishnan in \cite{10.1007/11780823_22}.  As we shall see, these results will ultimately allow us to recover good best $s$-term approximations to the eigenvectors $\u_j$ of $\A$ in sub-linear time using the eigenvectors $\utilde_j$ of $MAM^*$.  To begin we will review some basic definitions and supporting results.

\begin{definition}
	Let $ K, \alpha \in [N] := \{ 1, \ldots, N \} $.
	A matrix $ W \in \{0, 1\}^{ m \times N } $ is \emph{$ (K, \alpha) $-coherent} if the following conditions hold:
	\begin{enumerate}
		\item Every column of $ W $ contains at least $ K $ ones, and
		\item For every $ j, \ell \in [N] $, $ j \neq \ell $, the inner product of the columns $ {\bf w }_j $ and $ {\bf w }_\ell $ satisfies $ \langle {\bf w }_j, {\bf w }_\ell \rangle \leq \alpha $.
	\end{enumerate}
\end{definition}

Random constructions of $(K, \alpha) $-coherent matrices with a small number of rows include, e.g., this one from \cite[Theorem 2]{iwen2014}:

\begin{runex}
	A random matrix $ W \in \{0, 1\}^{ m \times N } $ with  i.i.d.\ Bernoulli random entries will be $ (K, \alpha) $-coherent with high probability under mild assumptions provided that $ m \geq c K^2 / \alpha $.
\end{runex}

There are also low-memory explicit and deterministic constructions with nearly as few rows, including those by, e.g., \cite{Kashin1975,DeVore2007,Iwen2008,Iwen2009}.  More generally, one can easily prove the following lemma after recalling a couple basic definitions from the theory of error correcting codes.

\begin{runex}
\label{ex:ECCconnection}
Let ${\bf c}_j, {\bf c}_\ell \in \{ 0,1 \}^m$. The Hamming weight of ${\bf c}_j$ is ${\rm wt}({\bf c}_j) := \|{\bf c}_j\|_1$. Moreover, the Hamming distance between ${\bf c}_j$ and ${\bf c}_\ell$ is ${\rm d}({\bf c}_j, {\bf c}_\ell) := \| {\bf c}_j - {\bf c}_\ell \|_1$. Any error correcting code $({\bf c}_0, \dots, {\bf c}_{N-1}) \in \{ 0,1 \}^{m \times N}$ with constant Hamming weight $K = {\rm wt}({\bf c}_j)$~$\forall j \in [N]$ and minimum Hamming distance $\Delta := \min_{j \neq \ell }\| {\bf c}_j - {\bf c}_\ell \|_1$ is also a $(K, K - \frac{\Delta}{2})$-coherent matrix. In fact, one can see that $\langle {\bf c}_j, {\bf c}_\ell \rangle = K - \frac{\| {\bf c}_j - {\bf c}_\ell \|_1}{2} \leq K - \frac{\Delta}{2}$ for all $j \neq \ell$.
\end{runex}

Given Example~\ref{ex:ECCconnection} one can see that there are in fact many deterministic constructions of $ (K, \alpha) $-coherent matrices waiting in the error correcting code literature.  See, e.g., \cite{lau2021construction} for more related discussion.  In addition to $ (K, \alpha) $-coherent matrices we will also need small error correcting code matrices known as ``bit testing matrices".  An example of such a matrix follows.

\begin{definition}[Bit Testing Matrices] \label{Def:BitTest}
    For $N \in \mathbbm{N}$ the $N^{\rm th}$ bit testing matrix, $B_N \in \{ 0, 1 \}^{(1 + \lceil \log_2 N \rceil ) \times N}$, is the matrix whose $j^{\rm th}$-column $\forall j \in [N]$ is a $1$ followed by $j-1$ written in binary.  For example,
    	\begin{equation*}
		B_8 =
		\begin{pmatrix}
			1 & 1 & 1 & 1 & 1 & 1 & 1 & 1\\
			0 & 1 & 0 & 1 & 0
            & 1 & 0 & 1\\
			0 & 0 & 1 & 1 & 0
            & 0 & 1 & 1\\
			0 & 0 & 0 & 0 & 1
            & 1 & 1 & 1       
		\end{pmatrix}.
	\end{equation*}
    Finally, below we will also consider $B_N' \in \{ 0, 1 \}^{2(1 + \lceil \log_2 N \rceil ) \times N}$, an extended version of $B_N$ created by appending each column of $B_N$ with its binary complement.  Note that every column of $B_N'$ will contain exactly $1 + \lceil \log_2 N \rceil$ ones.
\end{definition}

Note that each $B_N$ matrix defined above is effectively the parity check matrix for a Hamming code (with an additional row of $1$'s appended at the top).  As a result, linear systems of the form $B_N {\bf x} = {\bf y}$ can be solved for all $1$-sparse ${\bf x} = \gamma {\bf e}_j$ in $\mathcal{O}(\log N)$-time.  The same will therefore be true of $B_N'$ since it contains $B_N$ as a sub-matrix.

Finally, the compressive sensing matrices utilized in \cite{10.1007/11780823_22,BaileyIwenSpencer2012,iwen2014} all (implicitly) utilize the Khatri–Rao product defined below.

\begin{definition}[Khatri–Rao Product]
\label{Def:Khatri–Rao}
Let $B \in \C^{m \times N}$ and $C \in \C^{p \times N}$. Their \textit{Khatri–Rao product}, $B \odot C \in \C^{mp \times N},$ is defined as
\begin{equation*}
		B \odot C =
		\begin{pmatrix}
			B_{0,0} C_{:,0}& B_{0,1} C_{:,1} & \hdots & B_{0,N-1} C_{:,N-1} \\
			B_{1,0} C_{:,0} & B_{1,1} C_{:,1} & \hdots & B_{1,N-1} C_{:,N-1} \\
			\vdots & \vdots & \ddots & \vdots \\
			B_{m-1,0} C_{:,0} & B_{m-1,1} C_{:,1} & \hdots & B_{N-1,N-1} C_{:,N-1}
		\end{pmatrix}.
	\end{equation*}
More precisely, for any $j \in [mp]$ we note that we may uniquely write $j = q m + r$ for some $q \in [p]$ and $r \in [m]$.  The $(j,k) \in [mp] \times [N]$ entry of $B \odot C$ is then $\left(B \odot C\right)_{j,k} = B_{r,k}C_{q,k}$.
\end{definition}

The following sublinear-time compressive sensing result follows from a straightforward combination of \cite[Theorem 5.1.27]{Iwennotes} with \cite[Theorem 5.1.35]{Iwennotes}.\footnote {Alternatively, one may combine \cite[Lemma 2]{BaileyIwenSpencer2012} with trivial modification of \cite[Theorem 4]{BaileyIwenSpencer2012}.}

\begin{theorem} \label{Thm:FinalsublinearCSRes}
Let $s, K, \alpha \in [N]$, $W \in \{ 0,1 \}^{m' \times N}$ be $(K \geq 4 \alpha s+1,\alpha)$-coherent, and $B_N' \in \{ 0,1 \}^{ 2(1 + \lceil \log_2 N \rceil) \times N}$ be an extended bit testing matrix.  Then, there exists a compressive sensing algorithm $f: \mathbbm{C}^{2(1 + \lceil \log_2 N \rceil)m'} \rightarrow \mathbbm{C}^N$ such that
	\begin{equation*} 
		\| {\bf x } - f\left(\left( W \odot B_N' \right) {\bf x } + {\bf n }  \right) \|_2 \leq \| {\bf x } - {\bf x }_{ 2s } \|_2 + 6(1+\sqrt{2}) \left(\frac{\| {\bf x } - {\bf x }_s \|_1}{\sqrt{ s }} + \sqrt{ s } \| {\bf n } \|_\infty \right)
	\end{equation*}
holds for all ${\bf x} \in \mathbbm{C}^N$ and ${\bf n} \in \C^{2(1 + \lceil \log_2 N \rceil)m'}$. Furthermore, evaluations of $f$ (which are $2s$-sparse) can always be computed in $\mathcal{O}\left(m' \log N \right)$-time.
\end{theorem}

Note $W \odot B_N'$ will be $\left(K(1 + \lceil \log_2 N \rceil),\alpha (1 + \lceil \log_2 N \rceil) \right)$-coherent whenever $W$ is $(K,\alpha)$-coherent.  Furthermore, if $W$ is $(K,\alpha)$-coherent with $K$ ones in every column then $\frac{1}{\sqrt{K}}W$ will have the RIP of order $\left(r',\frac{(r'-1)\alpha}{K} \right)$ (recall Definition~\ref{def:RIP})
\cite[Theorem 5]{BaileyIwenSpencer2012}.  Finally, let $K \in [N]$.  There is a $(K,\lfloor \log_{K} N \rfloor)$-coherent matrix $W \in \{ 0,1\}^{m' \times N}$ due to DeVore \cite{DeVore2007} that has both $(i)$ $K$ ones in every column, and $(ii)$ $m' = \mathcal{O}(K^2)$ rows.\footnote{This $W$ is effectively constructed using a Reed–Solomon code.}  Combining these facts together with Theorem~\ref{Thm:FinalsublinearCSRes} we obtain the following lemma.

\begin{lemma} \label{lem:MergedPropsLem}
Let $s, r', 1/\epsilon \in [N]$ and $B_N' \in \{ 0,1 \}^{ 2(1 + \lceil \log_2 N \rceil) \times N}$ be an extended bit testing matrix.  Let $K := \max \{4s, r'/\epsilon \} \left\lceil \log_{\max \left\{s, \frac{r'}{\epsilon} \right\}} N \right\rceil$.  There exists a $\left(K, \lfloor \log_{K} N \rfloor \right)$-coherent matrix $W \in \{ 0,1\}^{m' \times N}$ with $m' = \mathcal{O}\left( K^2 \right)$ that both $(i)$ satisfies the conditions of Theorem~\ref{Thm:FinalsublinearCSRes}, and $(ii)$ has the property that $\tilde{W} := \frac{1}{\sqrt{K(1 + \lceil \log_2 N \rceil)}}W \odot B_N' \in \mathbbm{R}^{2(1 + \lceil \log_2 N \rceil)m' \times N}$ has the RIP of order $\left(r', \epsilon \right)$.
\end{lemma}

Set $t := \max\{s,r'/\epsilon\}$.  Looking at Lemma~\ref{lem:MergedPropsLem} we can see that using DeVore's construction for $W \in \{ 0,1\}^{m' \times N}$ as done therein will yield a matrix $\tilde{W} = \frac{1}{\sqrt{K(1 + \lceil \log_2 N \rceil)}}W \odot B_N' \in \mathbbm{R}^{m \times N}$ with 
$$m ~=~ \mathcal{O}(K^2 \log N) ~=~ \mathcal{O}\left( t^2 \log_t^2 N \log N\right)$$
rows that has the RIP of order $\left(r', \epsilon \right)$.  Furthermore, let $f: \mathbbm{C}^{2(1 + \lceil \log_2 N \rceil)m'} \rightarrow \mathbbm{C}^N$ be as per Theorem~\ref{Thm:FinalsublinearCSRes}, and set 
\begin{equation} \label{equ:RIPtypemeasurements}
\utilde ~:=~ \tilde{W} {\bf \u} + {\bf n} ~=~ W \odot B_N'\left( \frac{\bf \u }{\sqrt{K(1 + \lceil \log_2 N \rceil)}} \right) + {\bf n}.
\end{equation}
In this case we will also have
	\begin{align}
\label{equ:SublinErrorGuarantee}
		&\left \| \u - \sqrt { K(1 + \lceil \log_2 N \rceil) } f \left ( \utilde \right ) \right \|_2 \\ &\quad \quad \leq \| \u - \u_{ 2s } \|_2 + 6(1+\sqrt{2}) \left( \frac{\| \u - \u_s \|_1}{\sqrt{ s }} + \sqrt{ s K(1 + \lceil \log_2 N \rceil)} ~\| {\bf n } \|_\infty \right) \nonumber\\
        &\quad \quad = \| \u - \u_{ 2s } \|_2 + 6(1+\sqrt{2}) \left( \frac{\| \u - \u_s \|_1}{\sqrt{ s }} + \beta_{m}({\bf n }) \| {\bf n } \|_2 \right), \nonumber
	\end{align}
where 
\begin{equation} \label{equ:DefofBeta_mn}
\beta_{m}({\bf n }) := \frac{\| {\bf n } \|_\infty \sqrt{ s K(1 + \lceil \log_2 N \rceil)}}{\| {\bf n } \|_2}.
\end{equation}

Using this computation together with several classical results from the compressive sensing literature we can now prove the following theorem.
\begin{theorem} \label{Thm:MainSublinearMatrixSetup}
    Let $\A\in\mathbbm{C}^{\N\times\N}$, and choose $r, s, 1/\epsilon \in[\N]$ and $p \in (0,1)$.  Set $t := \max \left\{ 4s, \frac{512 r \lceil \ln(231 N / p) \rceil}{\epsilon}  \right\}$ and $K := t \left\lceil \log_{t} N \right\rceil$.  Let $B_N' \in \{ 0,1 \}^{ 2(1 + \lceil \log_2 N \rceil) \times N}$ be an extended bit testing matrix, and $W \in \{ 0,1\}^{\mathcal{O}\left( K^2 \right) \times N}$ be the $\left(K, \lfloor \log_{K} N \rfloor \right)$-coherent matrix from Lemma~\ref{lem:MergedPropsLem}.  Finally, let $D \in \mathbbm{R}^{N \times N}$ be a random diagonal matrix with i.i.d. Rademacher random variables on its diagonal, and set
    \begin{equation*}
    M := \tilde{W} D = \frac{1}{\sqrt{K(1 + \lceil \log_2 N \rceil)}} \left(W \odot B_N' \right)D \in \mathbbm{R}^{m \times N},
    \end{equation*}
    where $\tilde{W}$ is also as in Lemma~\ref{lem:MergedPropsLem}, and $m = \mathcal{O}(K^2 \log N)$.  
    
    Then, $M$ both $(i)$ satisfies all 3 assumptions of Theorem \ref{theorem:MAM} with probability at least $1-p$, and $(ii)$ satisfies
\begin{align*}
		&\left \| \u - \sqrt { K(1 + \lceil \log_2 N \rceil) } \cdot D f \left ( M \u + {\bf n } \right ) \right \|_2 \\ &\qquad \qquad \qquad \qquad \leq \| \u - \u_{ 2s } \|_2 + 6(1+\sqrt{2}) \left( \frac{\| \u - \u_s \|_1}{\sqrt{ s }} + \beta_{m}({\bf n }) \| {\bf n } \|_2 \right) 
\end{align*}
for all $\u \in \mathbbm{C}^N$ and ${\bf n} \in \mathbbm{C}^m$.  Here $\beta_{m}({\bf n })$ is as in 
\eqref{equ:DefofBeta_mn}, and $f: \mathbbm{C}^{m} \rightarrow \mathbbm{C}^N$ is as in Theorem~\ref{Thm:FinalsublinearCSRes} so that evaluations of $f$ can be computed in $\mathcal{O}\left(m \right)$-time.
\end{theorem}

\begin{remark}\label{remark:sublinearmbound}
The upper bound on $m$ provided by Theorem~\ref{Thm:MainSublinearMatrixSetup} can be simplified in terms of $r,s,\epsilon$ and $p$ as follows.
\begin{align*}
m = \mathcal{O}(K^2 \log N) = \mathcal{O}\left( t^2 \log_t^2 N \log N\right) = \mathcal{O}\left( \max \left\{ s^2, \frac{r^2}{\epsilon^2}  \right\} \log^5(N/p) \right).
\end{align*}
Though somewhat loose, we will utilize the rightmost upper bound outside of this section for simplicity.
\end{remark}

\begin{remark}
The matrix $M$ in provided by Theorem~\ref{Thm:MainSublinearMatrixSetup} is in fact somewhat sparse, containing only $K(1 + \lceil \log_2 N \rceil)$ nonzero entries in each column.  As a consequence, it can be multiplied by an arbitrary vector in better-than-expected $\mathcal{O}(N K \log N)$-time.  In addition, the highly structured nature of $M$ makes it extremely efficient to store and generate on the fly.  See, e.g., \cite[Section 2.1]{iwen2014} for a related discussion regarding how to compactly store and quickly generate portions of $W$ on demand.
\end{remark}

\begin{remark}\label{remark:sublinearBetaissmall}
In the application we consider herein it's reasonable to expect the noise vector ${\bf n} \in \mathbbm{C}^m$ appearing in Theorem~\ref{Thm:MainSublinearMatrixSetup} to be relatively {\it flat}, i.e., to satisfy $\| {\bf n} \|_2 \sim \sqrt{m} \| {\bf n} \|_\infty$.  Whenever this happens we can see from \eqref{equ:DefofBeta_mn} that 
$$\beta_{m}({\bf n }) \sim \sqrt{ \frac{s K(1 + \lceil \log_2 N \rceil)}{m}} = \mathcal{O}\left( \sqrt{\frac{s}{K}}\right) = \mathcal{O}\left( \sqrt{\frac{1}{\log_t N}}\right).$$
Hence, it's reasonable to expect that $\beta_{m}({\bf n })$ is generally $\mathcal{O}(1)$ herein.  See also Section~\ref{sec:Experiments} for empirical verification that this indeed does generally hold in practice.
\end{remark}

\begin{proof}
We first prove part $(i)$ of Theorem~\ref{Thm:MainSublinearMatrixSetup} by showing that $\M = \tilde{W} D \in\mathbbm{R}^{\m\times\N}$ satisfies each of the following properties with probability at least $1 - (p/3)$:
    \begin{enumerate}
        \item $\M$ is an $\epsilon$-JL map of the column space of $\A_\r$ into $\mathbbm{C}^\m$.
        \item $\M$ is an $\epsilon$-JL map of the column space of $\A_\r^*$ into $\mathbbm{C}^\m$.
        \item  M is an $\epsilon$-JL map of the smallest $N-r$ right and $N-r$ left singular vectors of A.
    \end{enumerate}
Claim $(i)$ then follows from the union bound.

Considering properties 1 and 2 above, one can appeal to, e.g., \cite[Lemma 3]{iwen2021lower} to see that $M$ will be an $\epsilon$-JL map of a given $r$-dimensional subspace $\mathcal{L}$ of $\mathbbm{C}^N$ if it's an $(\epsilon/2)$-JL map of an $(\epsilon / 16)$-net $\mathcal{C}$ of the unit $\ell_2$-sphere in $\mathcal{L}$.  Furthermore, this net $\mathcal{C}$ may be chosen so that $|\mathcal{C}| \leq \left(\frac{47}{\epsilon} \right)^{2r}$.  Applying \cite[Theorem 9.36]{FoucartRauhut2013} one can then further see that it suffices, in turn, for $\tilde{W}$ from Lemma~\ref{lem:MergedPropsLem} to have the RIP of order $\left( 64 r \ln(231 / \epsilon p) , \epsilon/8 \right)$ in order for $M = \tilde{W}D$ to be an $(\epsilon/2)$-JL map of an $(\epsilon / 16)$-net $\mathcal{C}$ of that size with probability at least $1 - (p/3)$.  Considering property 3 above, one can similarly see that it suffices for $\tilde{W}$ from Lemma~\ref{lem:MergedPropsLem} to have the RIP of order $\left( 32 \ln (48(N-r)/p), \epsilon/4 \right)$ in order for $\tilde{W}D$ to be an $\epsilon$-JL map of the $2(N-r)$ smallest right and left singular vectors of $A$ with probability at least $1 - (p/3)$.

Hence, in order for $M$ to satisfy all of properties $1 - 3$ with probability at least $1-p$ it suffices to set $\epsilon$ and $r'$ in Lemma~\ref{lem:MergedPropsLem} to be $\epsilon / 8$ and $r' = 64 \max \big\{ r \ln(231 / \epsilon p),$ $ \ln (48(N-r)/p) \big\} \leq 64 r \ln(231 N / p)$, respectively, where we have used that $\epsilon \geq 1/N$ to help more simply bound $r'$.  Doing so then guarantees that $\tilde{W}$ with $K$ defined as above will have all of the necessary RIP conditions required in order to guarantee that $M = \tilde{W} D$ will simultaneously satisfy all the assumptions of Theorem \ref{theorem:MAM} with probability at least $1-p$.   

Next, we prove part $(ii)$ by using \eqref{equ:SublinErrorGuarantee} in combination with the many special properties of any possible realization of the random diagonal matrix $D \in \mathbbm{R}^{N \times N}$.  Let $\u \in \mathbbm{C}^N$.  Note that $\| D \u \|_p = \| \u \|_p$ for all $p \geq 1$, $(D\u)_s = D (\u_s)$ for all $s \in [N]$, and $D^2 = I$ always hold. Using these properties of $D$ together with \eqref{equ:SublinErrorGuarantee} we have that
\begin{align*}
		&\left \| \u - \sqrt { K(1 + \lceil \log_2 N \rceil) } \cdot D f \left ( M \u + {\bf n } \right ) \right \|_2\\ &\quad \quad \quad \quad = \left \| D\u - \sqrt { K(1 + \lceil \log_2 N \rceil) } \cdot f \left ( \tilde{W}D \u + {\bf n } \right ) \right \|_2 \\ &\quad \quad \quad \quad \leq \| D\u - (D\u)_{ 2s } \|_2 + 6(1+\sqrt{2}) \left( \frac{\| D\u - (D\u)_s \|_1}{\sqrt{ s }} + \beta_{m}({\bf n }) \| {\bf n } \|_2 \right)\\ &\quad \quad \quad \quad = \| \u - \u_{ 2s } \|_2 + 6(1+\sqrt{2}) \left( \frac{\| \u - \u_s \|_1}{\sqrt{ s }} + \beta_{m}({\bf n }) \| {\bf n } \|_2 \right). 
\end{align*}
The result now follows from Theorem~\ref{Thm:FinalsublinearCSRes}.
\end{proof}

We will now provide an empirical demonstration of the sublinear-time compressive sensing strategy developed just above.  


\section{Preliminary Experiments on Large Matrices $A \in \mathbbm{R}^{N \times N}$} \label{sec:Experiments}

In this section we assess a sublinear-time-in-$N$ implementation of Algorithm~\ref{Alg:ApproxfromMAM} on large matrices $A$ (see Theorem~\ref{THM:MAINRESULT_SUBLINEAR} for related theory).  In all experiments below the real-valued test matrices $A \in \mathbbm{R}^{N \times N}$ are randomly generated symmetric, positive semidefinite, and exactly rank-$r$ matrices whose eigenvectors are exactly $s$-sparse.  More specifically, each individual test matrix $A \in \mathbbm{R}^{N \times N}$ is randomly generated in a sparse $({\rm row},{\rm column},{\rm value}) \in [N]^2\times \mathbbm{R}$ coordinate list format by
\begin{enumerate}
\item choosing $rs$ support values $s_1, \dots, s_{rs} \in [N]$ uniformly at random with replacement,

\item independently drawing $rs$ mean $0$ and variance $1$ normal random values $u_{1}, \dots, u_{rs} \in \mathbbm{R}$ to use as eigenvector entries, 

\item creating $r$ $\ell_2$-normalized and $s$-sparse eigenvectors $\u_1, \dots, \u_r \in \mathbbm{R}^N$ in a sparse $({\rm entry},{\rm value }) \in [N] \times \mathbbm{R}$ coordinate list format using the random values above, and then by

\item setting $A =\sum_{j=1}^{r} 2^{-j} \u_j \u_j^\top$
\end{enumerate}
for varying choices of $r,s \in [N]$.  We used $N = 2^{27}-1 \approx 130$ million in all experiments below (so that all matrices have $N^2 \sim 10^{16}$ entries that must be considered).  Finally, every value plotted in a figure below is the average of 100 experimental runs using 100 i.i.d.\ random test matrices $A_1, \dots, A_{100}$ generated as above.  

The measurement matrices $M \in \mathbbm{R}^{m \times N}$ used for all experiments are of the form
    \begin{equation}\label{Meas_structure}
    M = \frac{1}{\sqrt{K(1 + \lceil \log_2 N \rceil)}} \left(W \odot B_N' \right)D \in \mathbbm{R}^{m \times N},
    \end{equation}
where $W$ is a $\left(K, \alpha \right)$-coherent matrix from \cite{Iwen2009,BaileyIwenSpencer2012}, $B_N'$ is an extended bit testing matrix as per Definition~\ref{Def:BitTest}, and $D$ is a random diagonal matrix with i.i.d.\ Rademacher random variables on its diagonal.  To save working memory the matrix $M$ is not explicitly constructed in the experiments below. Instead, $M$ is implicitly represented by an $\mathcal{O}(N)$-memory\footnote{If one doesn't store the random diagonal matrix $D$ in working memory the algorithm can in fact be implemented to utilize only $\mathcal{O}(m)$-memory.} algorithm which computes matrix-vector products $M {\bf v} \in \mathbbm{C}^m$ for any given vector ${\bf v} \in \mathbbm{C}^N$ when given sampling access to ${\bf v}$'s entries.  The sketch $MAM^* \in \mathbbm{R}^{m \times m}$ is then computed using $$MAM^* = \sum_{j=1}^{r}2^{-j}(M\u_j)(M\u_j)^\top.$$ The Compressive Sensing (CS) algorithm $\mathcal{A}: \mathbbm{C}^m \rightarrow \mathbbm{C}^N$ used in our implementation of Algorithm~\ref{Alg:ApproxfromMAM} was \cite[Algorithm 25]{Iwennotes}.\footnote{Note that \cite[Algorithm 25]{Iwennotes} is a trivial modification of \cite[Algorithm 1]{BaileyIwenSpencer2012}.}  The code used for all experiments below is publicly available at \url{https://github.com/boaed/MAM-Method-for-Eigenvector-Recovery}.

\begin{remark}Note that the test matrices $A$ constructed in this section don't quite satisfy all of the assumptions of Theorems \ref{THM:MAINRESULT_LINEARcosamp} and \ref{THM:MAINRESULT_SUBLINEAR}.   In particular, the test matrices herein have eigenvalues $\lambda_j = 2^{-j}$ for $j \in [r]$ (i.e., they have $c=1$, $q = 1/2$ and $\epsilon = 0$ since $\|A_{{\setminus}r}\|_* = 0$ ) which violates the condition $q < 1/3$ required by Theorem \ref{MainTHM:ExpDecayingSingVlaues} on which both Theorems \ref{THM:MAINRESULT_LINEARcosamp} and \ref{THM:MAINRESULT_SUBLINEAR} rely.  That said, the violated restriction to $q < 1/3$ was primarily made for cosmetic reasons in order to help relate the quantities $g_j, \nu_j$ and $\kappa_j$ in Corollary~\ref{Coro:approxerrorinEigenvectors} back to the spectrum of $A$ in an easily interpretable fashion.  Experiments indicate that the method performs well numerically under far more general conditions than those highlighted in Theorems \ref{THM:MAINRESULT_LINEARcosamp} and \ref{THM:MAINRESULT_SUBLINEAR}.
\end{remark}

\subsection{An Evaluation of Algorithm~\ref{Alg:ApproxfromMAM} as Sparsity and Rank Vary}

In the experiments below we empirically evaluate the errors considered in Theorem~\ref{MainTHM:ExpDecayingSingVlaues} and Theorem~\ref{THM:MAINRESULT_SUBLINEAR} as the true eigenvector-sparsity $s$ and rank $r$ used to generate the $100$ random test matrices $A_1, \dots, A_{100} \in \mathbbm{R}^{(2^{27}-1) \times (2^{27}-1) }$ used for each plotted value vary.  In each plot below the size parameter $m$ of the $MAM^* \in \mathbbm{R}^{m \times m}$ sketch is held fixed.  Two sets of average errors are then plotted as a function of $s$ (with $r$ fixed) and $r$ (with $s$ fixed):
\begin{enumerate}
    \item \underline{Relative $\ell_2$-error before CS inversion:}  The error averaged over our $100$ trial matrices, $$\frac{1}{100} \sum_{k = 1}^{100} \frac{\min \left\{ \|\utilde^k_j - M\u^k_j\|_2,\|\utilde^k_j + M\u^k_j\|_2\right\}}{\| \utilde^k_j \|_2} = \frac{1}{100} \sum_{k = 1}^{100} \min \left\{ \|\utilde^k_j - M\u^k_j\|_2,\|\utilde^k_j + M\u^k_j\|_2\right\},$$
where $\utilde^k_j \in \mathbbm{R}^m$ is the eigenvector of $M A_k M^*$ associated with it's $j^{\rm th}$ largest eigenvalue, and $\u^k_j \in \mathbbm{R}^N$ is the eigenvector of $A_k$ associated with it's $j^{\rm th}$ largest eigenvalue.  As above, we always consider our eigenvectors to be $\ell_2$-normalized so that $\|\u^k_j\|_2 = \|\utilde^k_j\|_2 = 1$ for all $k \in [100]$ and $j \in [m]$.  This still leaves a sign indeterminacy in our eigenvectors, however, which is the reason for minimizing over $\|\utilde^k_j - M\u^k_j\|_2$ and $\|\utilde^k_j + M\u^k_j\|_2$ above.

\item \underline{Relative $\ell_2$-error after CS inversion:}  This error averaged over our $100$ trial matrices is $$\frac{1}{100} \sum_{k = 1}^{100} \frac{\min \left\{ \|\u^k_j - \mathcal{A}(\utilde^k_j)\|_2,\|\u^k_j + \mathcal{A}(\utilde^k_j)\|_2\right\}}{\| \u^k_j \|_2} = \frac{1}{100} \sum_{k = 1}^{100} \min \left\{ \|\u^k_j - \mathcal{A}(\utilde^k_j)\|_2,\|\u^k_j + \mathcal{A}(\utilde^k_j)\|_2\right\},$$
where $\utilde^k_j \in \mathbbm{R}^m$, $\u^k_j \in \mathbbm{R}^N$ are as above, and the CS algorithm $\mathcal{A}: \mathbbm{C}^m \rightarrow \mathbbm{C}^N$ is \cite[Algorithm 25]{Iwennotes}.
\end{enumerate}
Results for the top $4$ eigenvectors of our test matrices $A$ (i.e., for $j \in [4]$) are plotted in Figures~\ref{Fig:sparsityvaries} and~\ref{Fig:rankvaries}.

\begin{figure}[!htb]
	\centering
	\begin{minipage}[t]{0.48\textwidth}
		\centering
		\includegraphics[height=5.5cm]{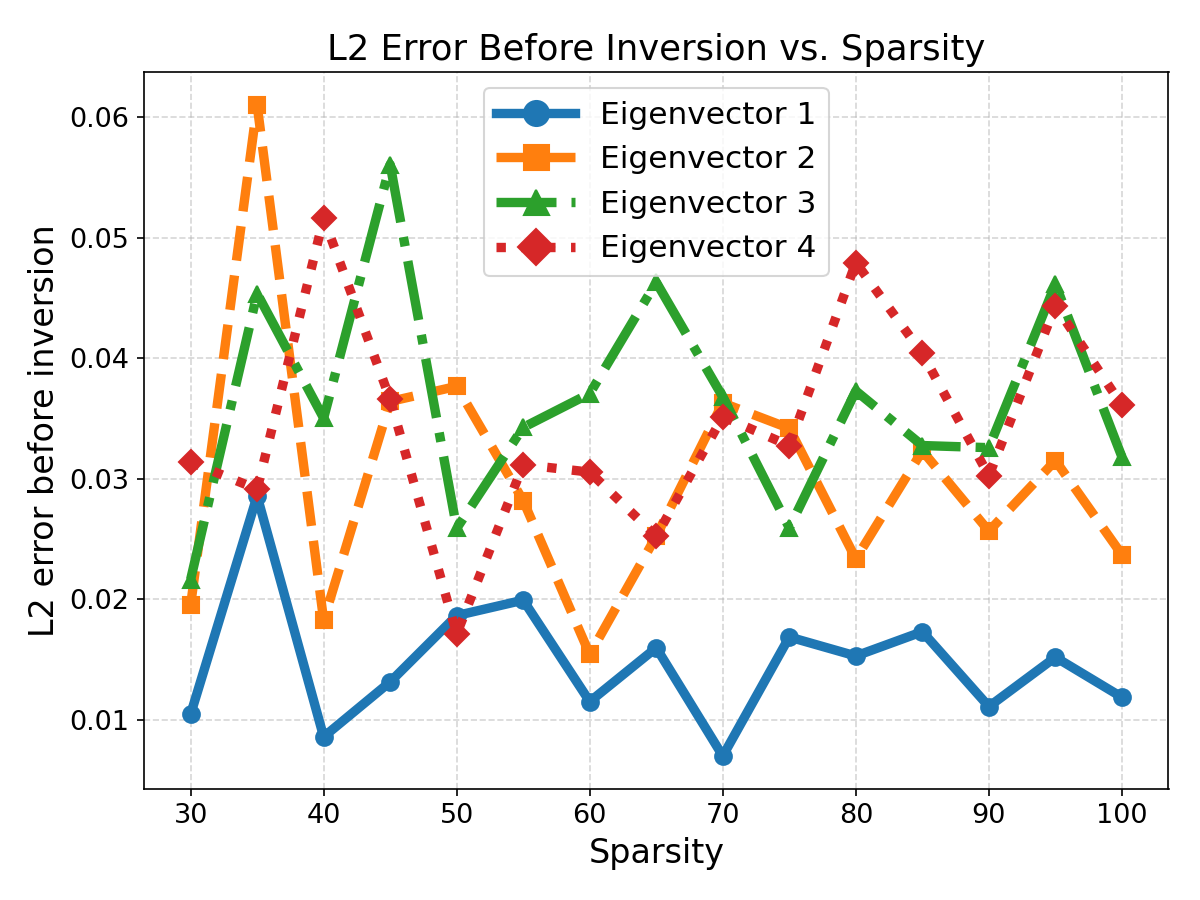}
		\caption*{$\min\{\|\tilde{u}_j - M u_j\|_2,\ \|\tilde{u}_j + M u_j\|_2\}$, $j \in [4]$}
	\end{minipage}%
	\hfill
	\begin{minipage}[t]{0.48\textwidth}
		\centering
		\includegraphics[height=5.5cm]{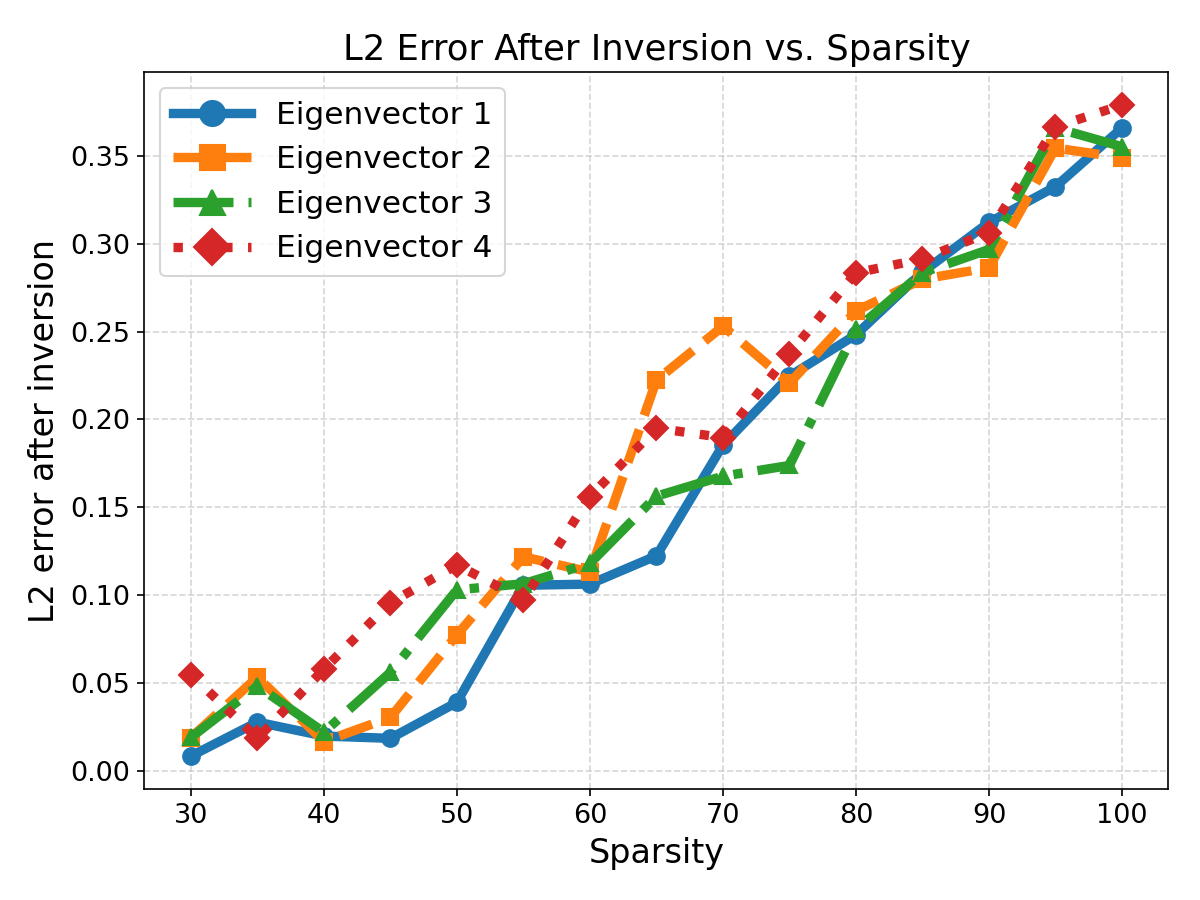}
		\caption*{$\min\{\|u_j - \mathcal{A}(\tilde{u}_j)\|_2,\ \|u_j + \mathcal{A}(\tilde{u}_j)\|_2\}$, $j \in [4]$}
	\end{minipage}
	
	\caption{
		Measurement and CS approximation errors for fixed rank $r=20$ matrices 
		$A \in \mathbb{R}^{(2^{27}-1)\times(2^{27}-1)}$ whose eigenvectors $u_j$ 
		are $s$-sparse for $s \in [30,100]$. 
		The measurement matrix $M \in \mathbb{R}^{m \times N}$ defined in 
		\ref{Meas_structure} had 
		$m = 2(1+\lceil\log_2(2^{27}-1)\rceil)\cdot 2147 = 120{,}232$, 
		so that $MAM^* \in \mathbb{R}^{120232 \times 120232}$.
	}
	\label{Fig:sparsityvaries}
\end{figure}

\begin{figure}[!htb]
	\centering
	\begin{minipage}[t]{0.48\textwidth}
		\centering
		\includegraphics[width=0.9\linewidth]{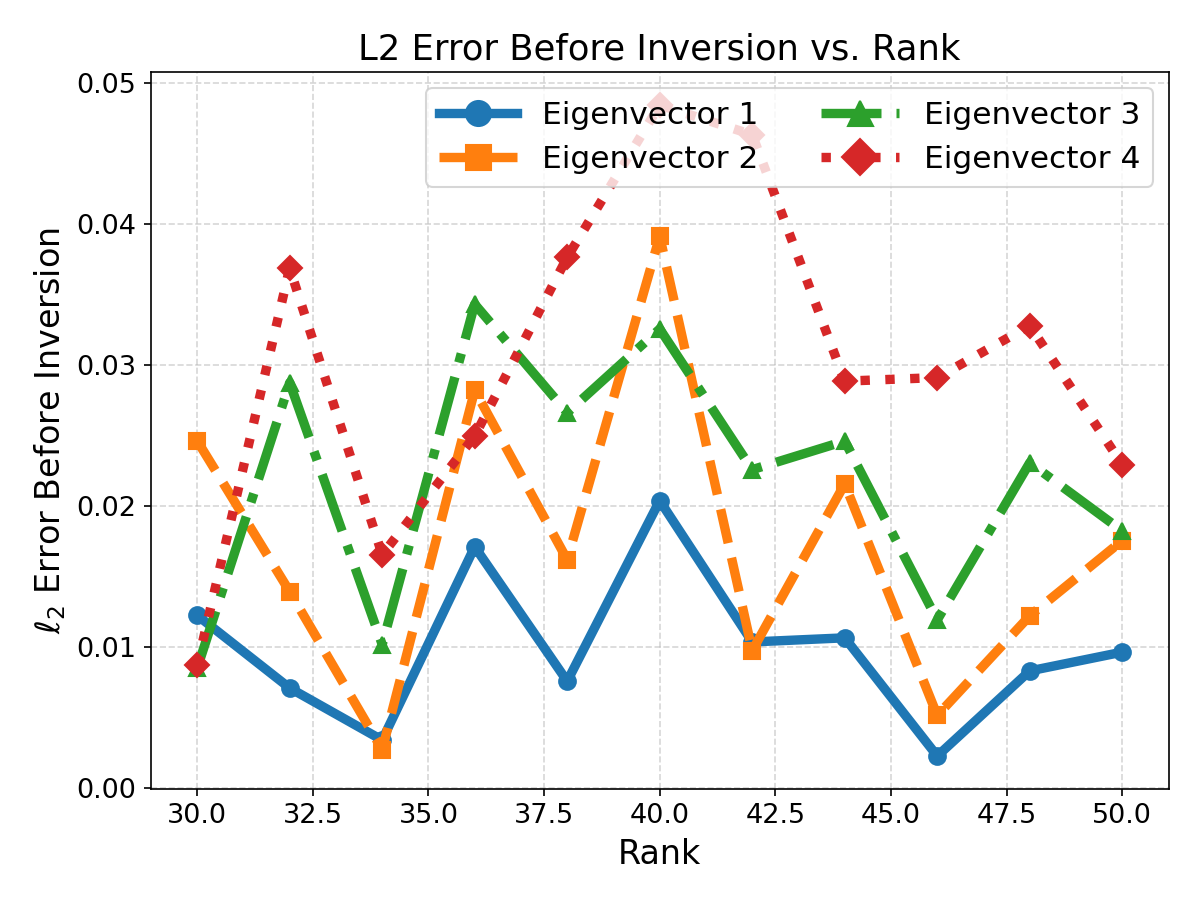}
		\caption*{$\min\{\|\tilde{u}_j - M u_j\|_2,\ \|\tilde{u}_j + M u_j\|_2\}$, $j \in [4]$}
	\end{minipage}%
	\hfill
	\begin{minipage}[t]{0.48\textwidth}
		\centering
		\includegraphics[width=0.9\linewidth]{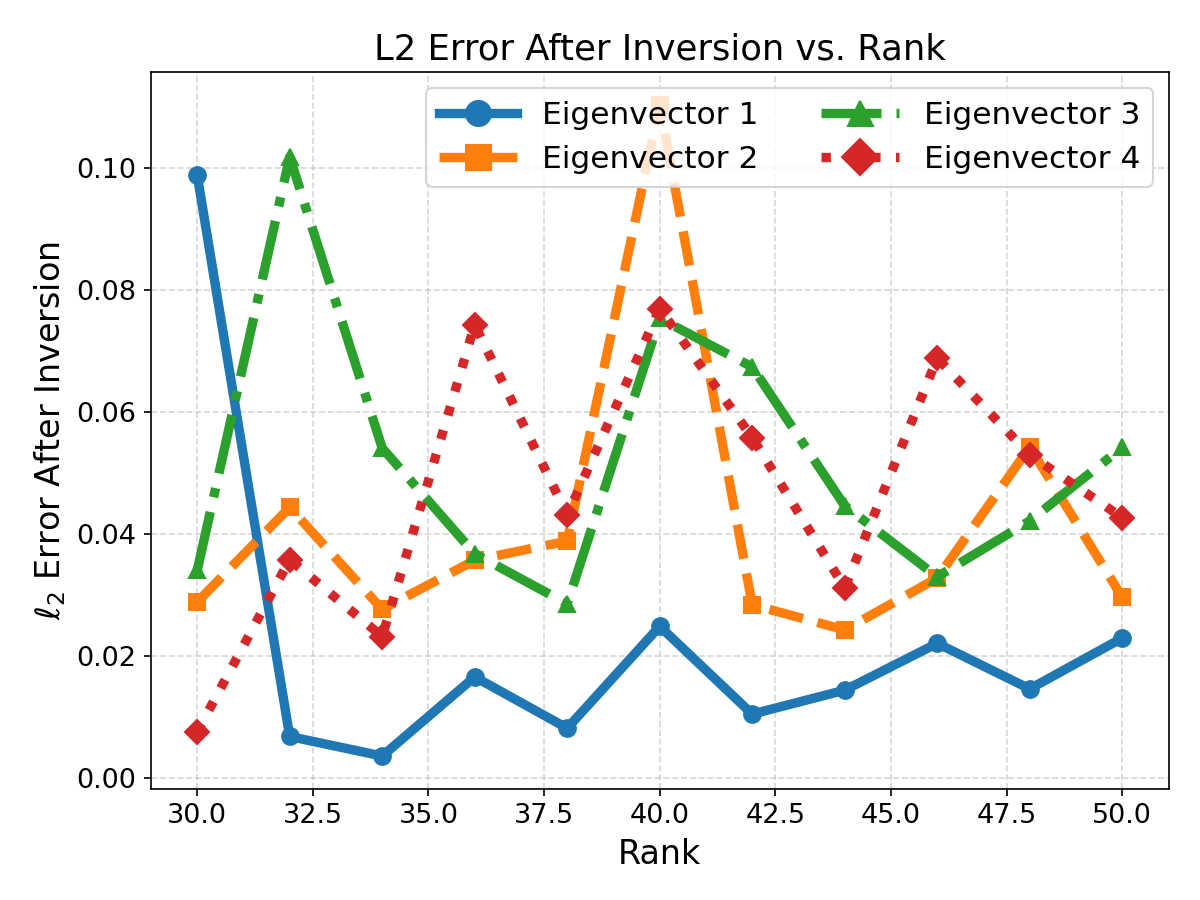}
		\caption*{$\min\{\|u_j - \mathcal{A}(\tilde{u}_j)\|_2,\ \|u_j + \mathcal{A}(\tilde{u}_j)\|_2\}$, $j \in [4]$}
	\end{minipage}
	
	\caption{
		Measurement and CS approximation errors for rank $r \in [30,50]$ matrices 
		$A \in \mathbb{R}^{(2^{27}-1)\times(2^{27}-1)}$ whose eigenvectors $u_j$ 
		are all fixed $s=100$ sparse. 
		The measurement matrix $M \in \mathbb{R}^{m \times N}$ had 
		$m = 2(1+\lceil\log_2(2^{27}-1)\rceil)\cdot 2965 = 166{,}040$, 
		so that $MAM^* \in \mathbb{R}^{166040 \times 166040}$.
	}
	\label{Fig:rankvaries}
\end{figure}

\paragraph{Effect of Eigenvector Sparsity on Reconstruction Error.}

Figure~\ref{Fig:sparsityvaries} (left) shows that the $\ell_2$ error before inversion remains low and consistent across a range of sparsity levels in the true eigenvectors, and behaves largely as expected with respect to the eigen-structure of $A$.  The CS measurements of the leading eigenvector of $A$ provided by the leading eigenvector of $MAM^*$ are generally most accurate, followed by the measurements of the second leading eigenvector of $A$ provided by the second leading eigenvector of $MAM^*$, etc..  This suggests that the sketching process retains the dominant eigenvectors' CS measurement information even as sparsity varies.  

On the other hand, the error after CS inversion (right) grows noticeably as the eigenvectors become less sparse. This trend indicates that standard CS measurement thresholds are at play.  That is, the way the CS measurements are being computed (as eigenvectors of the $MAM^*$ sketch) appears to be less of a bottleneck than the usual requirements of one's favorite CS algorithm.  We consider this to be a positive attribute of the proposed approach.

\paragraph{Effect of Matrix Rank on Reconstruction Error.}

As shown in Figure~\ref{Fig:rankvaries} (left), the error before inversion remains consistently small (below $0.05$ accross all runs) as the rank of the PSD matrix increases with other parameters held fixed. This suggests that the sketch captures the dominant eigenvector measurements well even as the number of less significant eigenvectors grows. Moreover, the error after inversion is relatively stable as a function of the rank.  This again bolsters our general observation that the recovery errors one will encounter using the proposed approach will be largely dominated by properties of the CS algorithm one employs for matrices $A$ that exhibit sufficiently fast spectral decay.

\subsection{An Evaluation of Algorithm~\ref{Alg:ApproxfromMAM} as the Sketch Size Varies}

In Figure~\ref{fig:Measurevary} $r$ is fixed to $20$, $s$ is fixed to $100$, and then the same two relative $\ell_2$ errors as in the last section are plotted as $m$ varies.  \\

\begin{figure}[!htb]
\centering
\begin{minipage}[t]{0.48\textwidth}
  \centering
  \includegraphics[width=\linewidth]{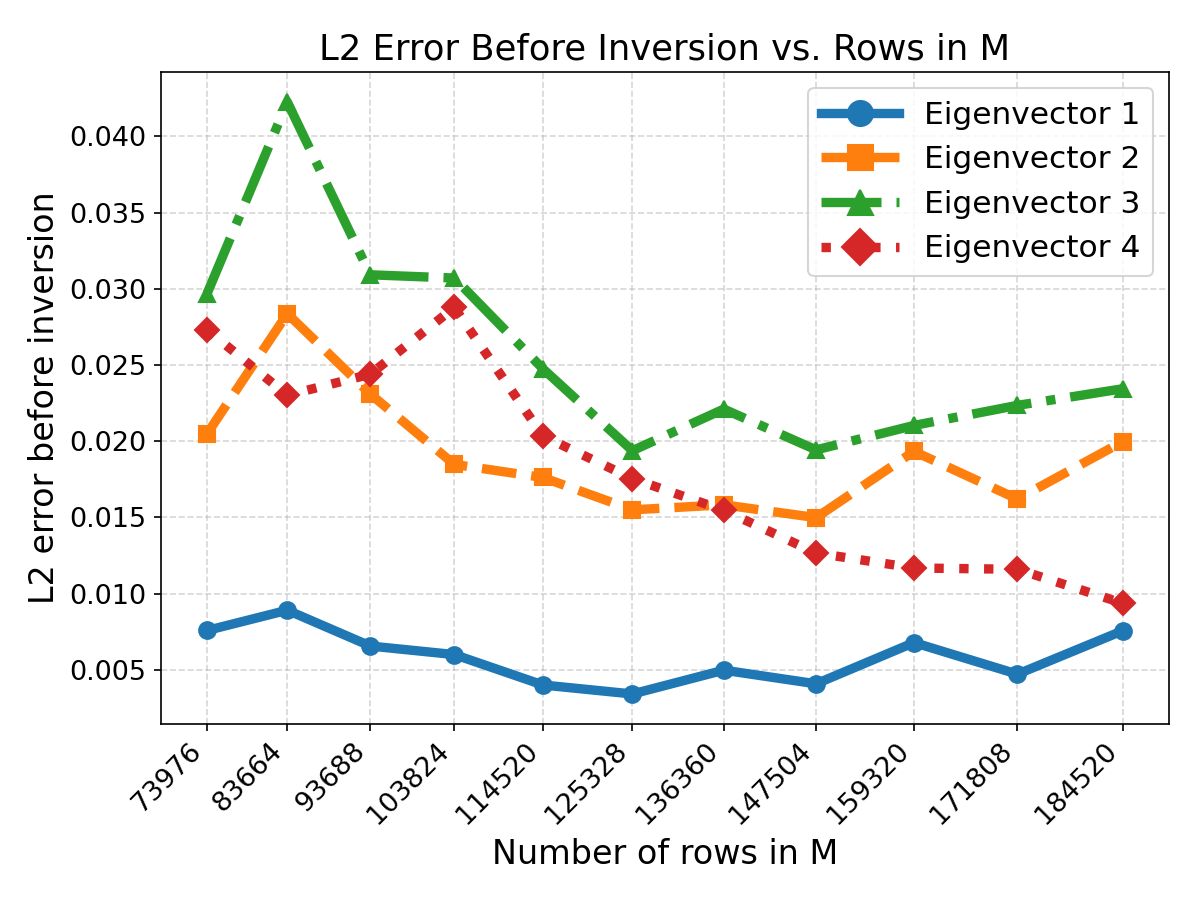}
  \captionof*{figure}{$\min \left\{ \|\tilde{ \u}_j - M \u_j\|_2,\|\tilde{\u}_j + M \u_j\|_2\right\}$, $j \in [4]$}
  \label{fig:beforeinversionmeasure}
\end{minipage}%
\hfill
\begin{minipage}[t]{0.48\textwidth}
  \centering
  \includegraphics[width=\linewidth]{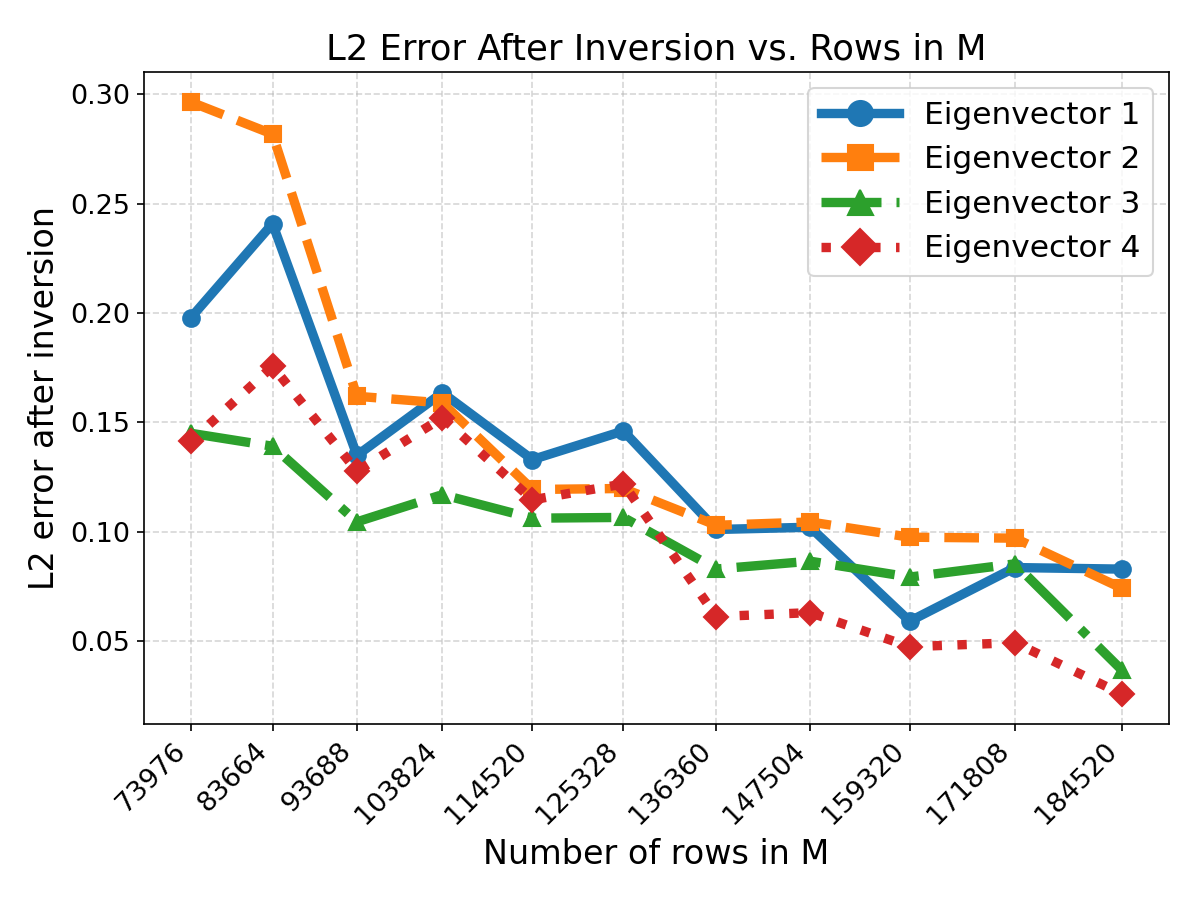}
  \captionof*{figure}{$\min \left\{ \| \u_j - \mathcal{A}(\tilde{\u}_j) \|_2, \| \u_j + \mathcal{A}(\tilde{\u}_j) \|_2 \right\}$, $j \in [4]$}
  \label{fig:afterCSerrormeas}
\end{minipage}
\caption{Measurement and CS approximation errors for rank $r =50$ matrices $A \in \mathbbm{R}^{(2^{27}-1) \times (2^{27}-1)}$ whose eigenvectors $\u_j$ are all $s = 100$ sparse.  The measurement matrices $M \in \mathbbm{R}^{m \times N}$ have $m \in [73976,184520]$ rows. Effectively this corresponds to the matrix $W$ with $w \in [1321,3295]$ rows.}
\label{fig:Measurevary}
\end{figure}
\paragraph{Effect of the Number of Measurements on Error.}

Figure~\ref{fig:Measurevary} (left) shows that the reconstruction error before inversion decreases slowly as the number of measurements $m$ increases, while the error after inversion (right) drops more sharply. Principally, with more measurements, the sketch contains richer information about the leading eigenvectors which then supports more reliable CS recovery. 

\subsection{The $\beta^{A}_{M}$ Value in Theorem~\ref{THM:MAINRESULT_SUBLINEAR}} \label{sec:BetaExperiments}

Recall that $\beta^{A}_{M} = 7 \max_{j \in [\ell]} \beta_{m}\left(\utilde_j - \mathbbm{e}^{\mathbbm{i} \phi'_j}M\u_j \right)$   , where $\beta_{m}(\cdot)$ is defined as in \eqref{equ:DefofBeta_mn} and $\phi'_j \in [0,2\pi)$ is such that $\left\| \utilde_j - \mathbbm{e}^{\mathbbm{i} \phi'_j} M \u_j \right\|_2 = \min_{\phi \in [0,2\pi)} \| \mathbbm{e}^{\mathbbm{i} \phi} \utilde_j - M \u_j \|_2$.  In this section explore empirical values of $\beta^{A}_{M}$ in real-valued setting when $\ell = 4$.  Toward this end we define $s_j^k \in \{-1,1\}$ to be such that 
$$\left\| \utilde_j^k - s_j^k M \u^k_j \right\|_2 = \min \left\{ \left\|\utilde_j^k - M\u_j^k \right\|_2, \left\|\utilde_j^k + M\u_j^k \right\|_2\right\},$$
and then plot
$$\beta = \frac{1}{100}\sum_{k=1}^{100} \max_{j \in [4]} \beta_{m}\left(\utilde^k_j - s_j^kM\u_j^k \right) = \frac{1}{100}\sum_{k=1}^{100} \left(\max_{j \in [4]} \frac{ \left\| \utilde^k_j - s_j^kM\u_j^k \right\|_\infty \sqrt{ s K(1 + \lceil \log_2 N \rceil)}}{\left\| \utilde_j^k - s_j^kM\u_j \right\|_2}\right),$$  
where $\utilde^k_j \in \mathbbm{R}^m$ and $\u^k_j \in \mathbbm{R}^N$ are as above.  See Figure~\ref{fig:beta_m} for plots $\beta$ as $s$, $r$, and $m$ each vary.  Recall that these values should be $\mathcal{O}(1)$ whenever $\utilde_j^k - s_j^k M \u^k_j$ is flat (see Remark~\ref{remark:sublinearBetaissmall}).

\begin{figure}[!htb]
\centering
\begin{minipage}{.33\textwidth}
  \centering
  \includegraphics[width=\linewidth]{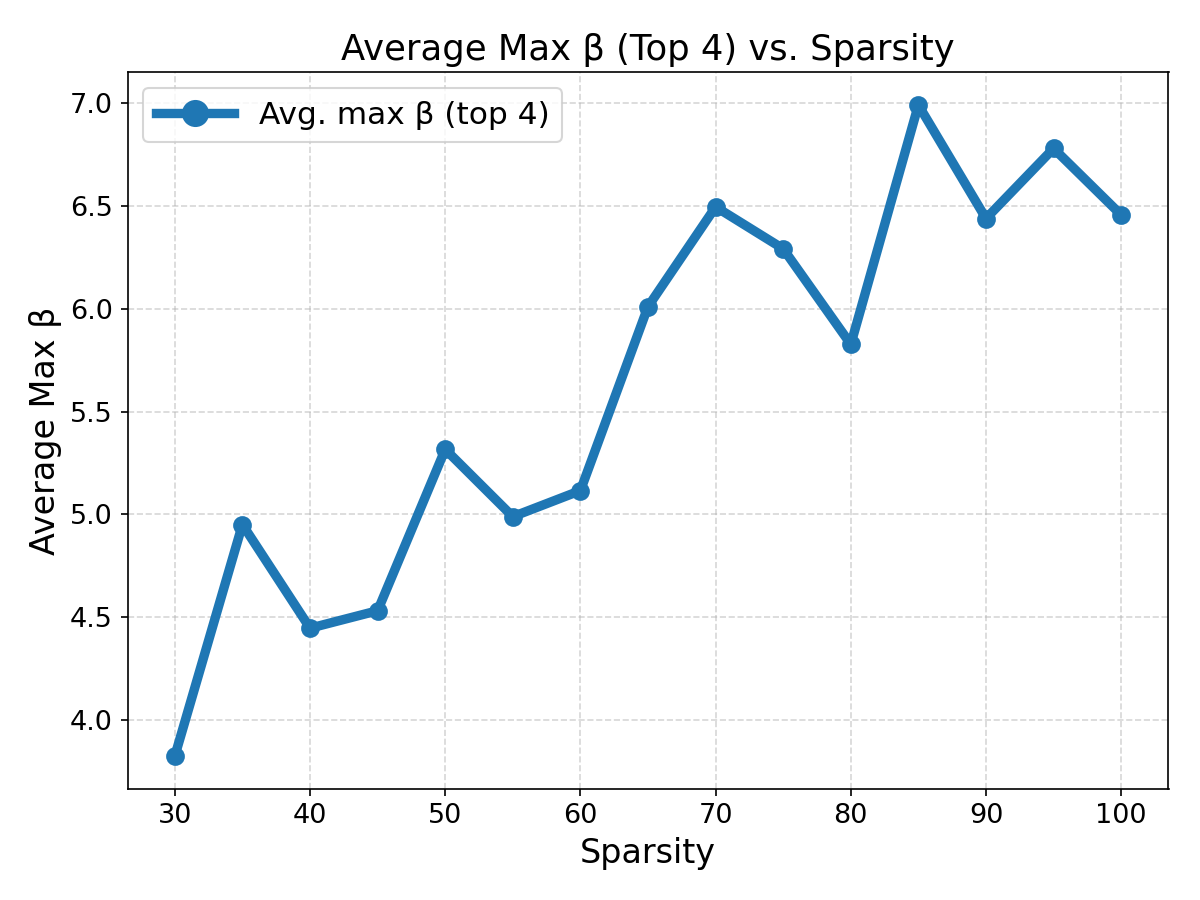}
  \captionof*{figure}{\tiny $r = 20$, $m = 120232$, $s \in [30,70]$}
  \label{fig:test_beta1}
\end{minipage}%
\begin{minipage}{.33\textwidth}
  \centering
  \includegraphics[width=1\linewidth]{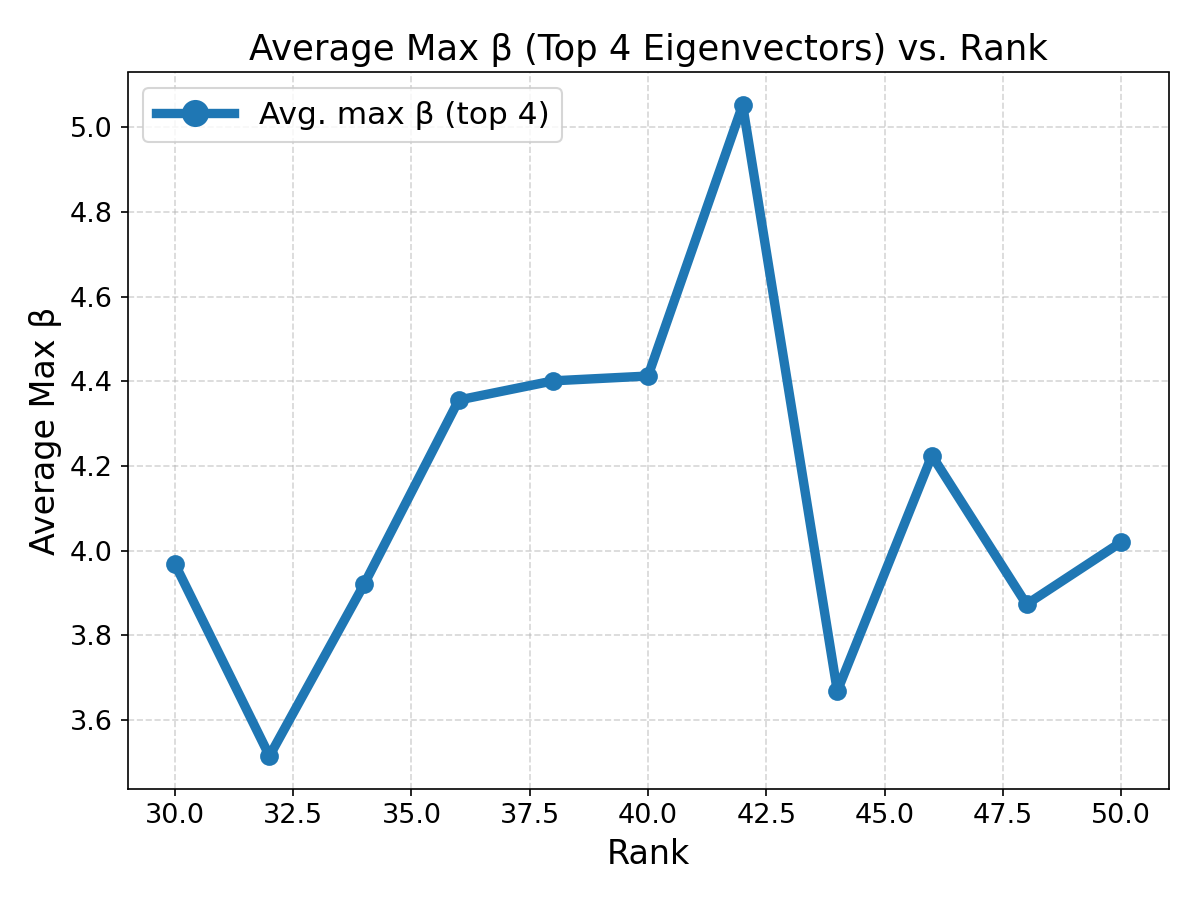}
  \captionof*{figure}{\tiny $s = 100$, $m = 166040$, $r \in [30, 50]$}
  \label{fig:test_beta2}
\end{minipage}
\begin{minipage}{.33\textwidth}
  \centering
  \includegraphics[width=\linewidth]{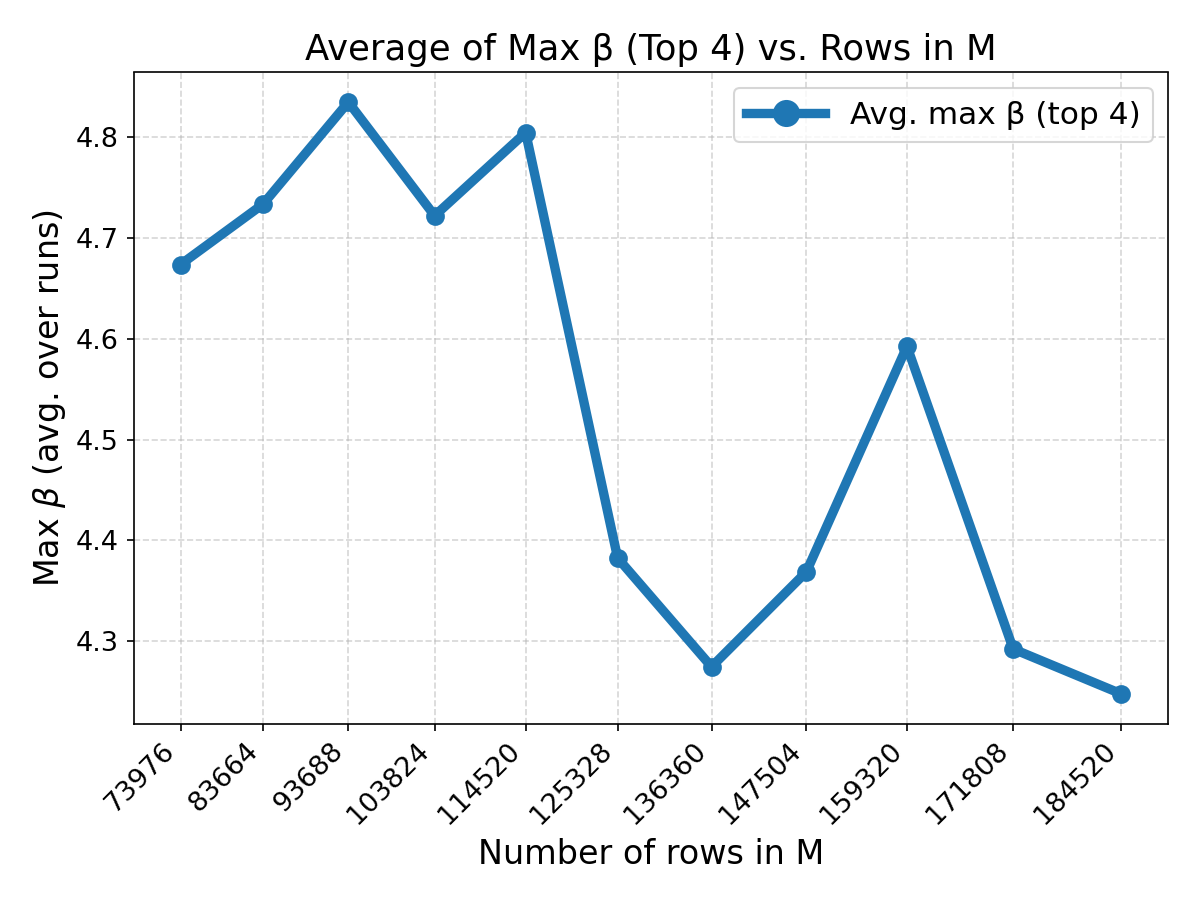}
  \captionof*{figure}{\tiny $r =50$, $s = 100$, $m \in [73976,184520]$}
  \label{fig:test_beta3}
\end{minipage}
\caption{Plot of average $\beta_m$ values associated with Theorem~\ref{THM:MAINRESULT_SUBLINEAR} across the experiments reported in Section~\ref{sec:Experiments}.  Note that all the reported values are less than $8$.}\label{fig:beta_m}
\end{figure}
\paragraph{Behavior of $\beta$.}

Figure~\ref{fig:beta_m} demonstrates that the average $\beta$ value remains less than $8$ in all the experiments presented herein.  Moreover, $\beta$ tends to decrease with more measurements (see, e.g., the far right plot) in keeping with our intuition.  

\subsection{Noisy Low-Rank Matrices with Compressible Eigenvectors and Slower Spectral Decay}

Let ${\bf u'} \in \mathbbm{R}^N$ for $N = 10^7$ have compressible entries given by
$$u'_k := \begin{cases} 1 & {\rm if}~ 1 \leq k \leq s := 100 \\
\frac{1}{(k-s+1)^2} & {\rm if}~ k > s 
\end{cases}, $$
and then set ${\bf u} := {\bf u'}/\| {\bf u'} \|_2$.  In this section we consider random rank-$4$ test matrices of the form 
\begin{align}
A_j := \left( \sum_{\ell=1}^{3} (11 - \ell) {\bf u}_{j,\ell} {\bf u}_{j,\ell}^{*} \right) + \tau {\bf g}_j {\bf g}_j^* \in \mathbbm{R}^{10^7 \times 10^7},  
\label{equ:compressibletrialAs}
\end{align} 
where each ${\bf g}_j \in \mathbbm{R}^N$ is an i.i.d. normalized standard Gaussian vector, and 
each ${\bf u}_{j,\ell} := \Pi_{j,\ell} D_{j,\ell}{\bf u}$ is independently generated using a uniformly random permutation matrix $\Pi_{j,\ell} \in \{0,1\}^{N \times N}$ together with a random diagonal sign matrix $D_{j,\ell} \in \{-1, 0, 1\}^{N \times N}$ having i.i.d. Rademachers on it's diagonal.  The set $S_{j,\ell} \subset [N]$ records the locations of the $s$ largest-magnitude entries in each ${\bf u}_{j,\ell} \in \mathbbm{R}^N$ so that $S_{j,\ell} := \left\{ i ~\big|~ |({\bf u}_{j,\ell})_i| = 1 \right\}$.  As in the preceding experiments, the measurement matrix is of the form \eqref{Meas_structure} where, here, 
$W$ has $7,013$ rows and $B_N'$ 
has $50$ rows, yielding $m = 50 \times 7,013 = 350,650 < N = 10^7$ total rows in $M$ for all experiments.

Below we consider $100$ $N \times N$ i.i.d. test matrices $A_1, \dots, A_{100}$ generated according to \eqref{equ:compressibletrialAs}.  Each matrix sketch $MA_jM^* \in \mathbbm{R}^{m \times m}$ is then computed as 
$$ MA_jM^* = \tau (M{\bf g}_j)(M{\bf g}_j)^* + \sum_{\ell = 1}^{3}(11-\ell)(M{\bf u}_{j,\ell})
(M{\bf u}_{j,\ell})^* ~~ \forall j \in [100],$$
where the value $\tau$ varies over the set $\{10^7, 10^8, \ldots, 10^{16}\}$.  Note that the dominant eigenvector of each sketch matrix $MA_jM^*$ corresponds to the rank-$1$ additive noise term $\tau (M{\bf g}_j)(M{\bf g}_j)^*$ since $\tau \gg 10$ always holds.  We denote the eigenvectors corresponding to the second, third, and forth largest eigenvalues of $MA_jM^*$ by $\tilde{\bf u}_{j,1} \in \mathbbm{R}^m, \tilde{\bf u}_{j,2} \in \mathbbm{R}^m$, and $\tilde{\bf u}_{j,3} \in \mathbbm{R}^m$, respectively.  The CS algorithm $\mathcal{A}: \mathbbm{C}^m \rightarrow \mathbbm{C}^N$ utilized below is again 
\cite[Algorithm~25]{Iwennotes}.\footnote{Recall that \cite[Algorithm 25]{Iwennotes} is a trivial modification of \cite[Algorithm 1]{BaileyIwenSpencer2012}.}

\begin{figure}[ht]
\centering
\includegraphics[width=0.32\textwidth]{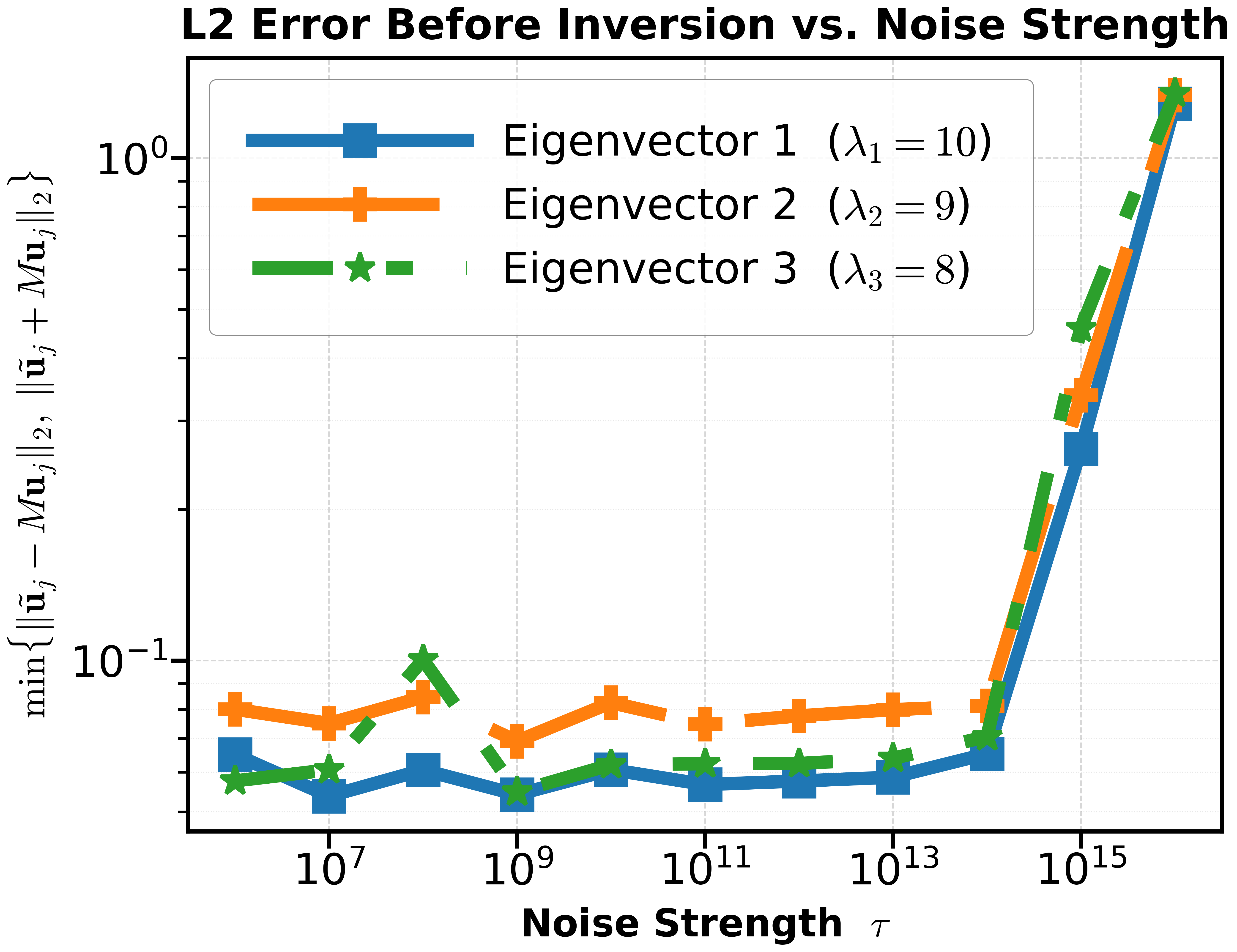}
\hfill
\includegraphics[width=0.32\textwidth]{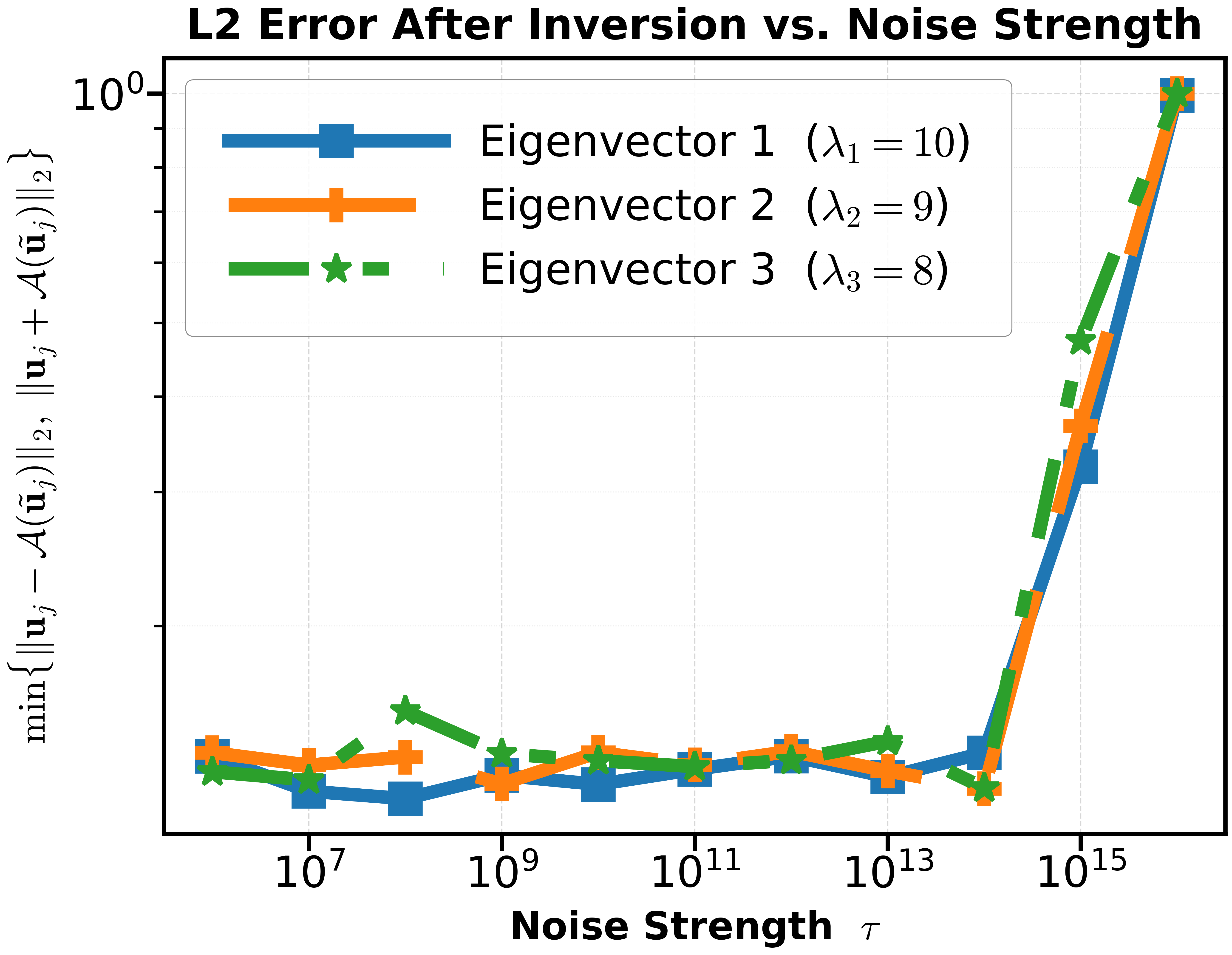}
\hfill
\includegraphics[width=0.32\textwidth]{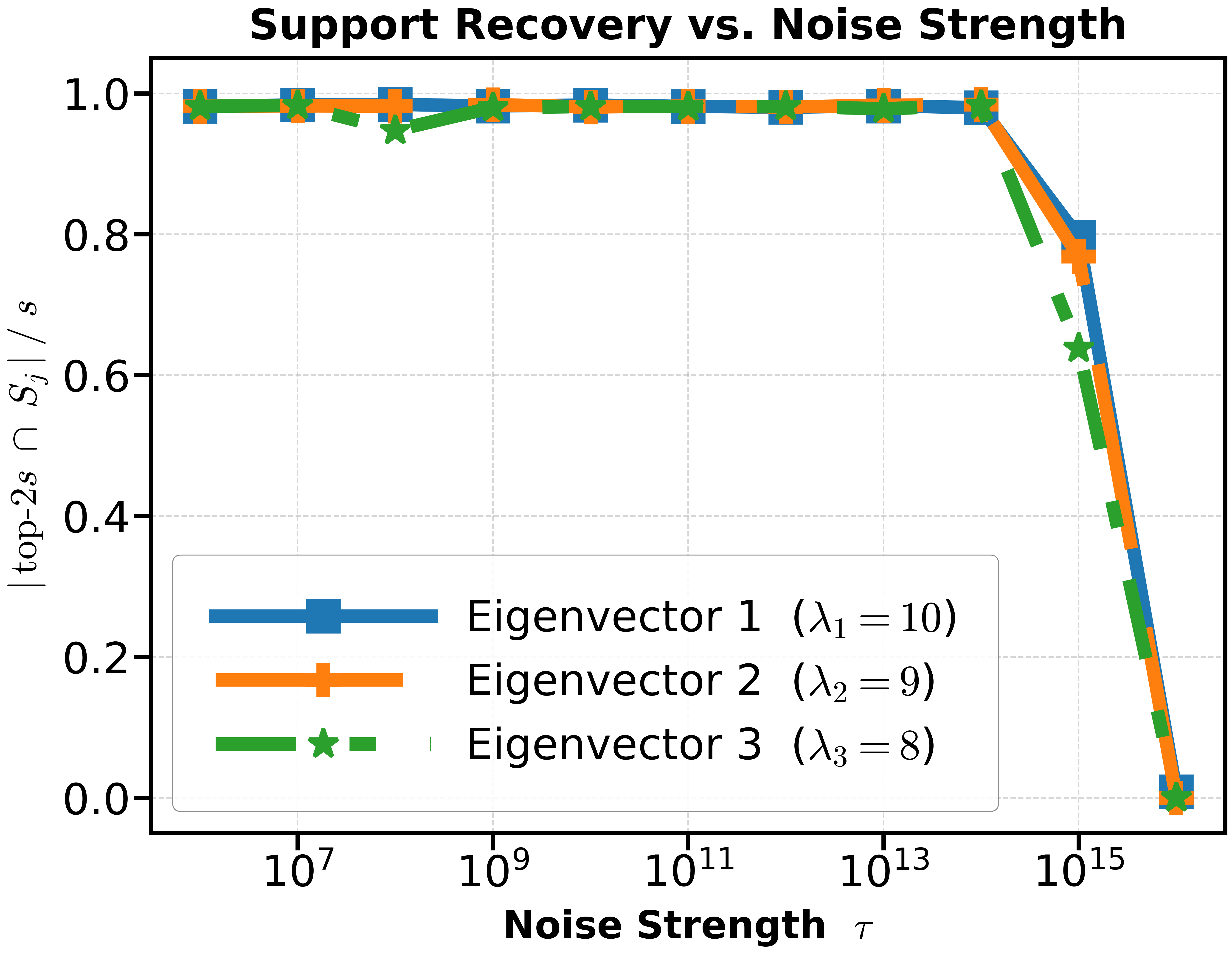}
\caption{Measurement error (left), CS recovery error (center), and
support recovery (right) for rank-$4$ matrices
$A\in\mathbb{R}^{(10^7)\times(10^7)}$ with compressible eigenvectors as per \eqref{equ:compressibletrialAs}
($s=100$, $\lambda_1 =10$, $\lambda_2=9$, $\lambda_3=8$)
corrupted by dominant rank-$1$ Gaussian noise $\tau\in[10^{7},10^{16}]$.}
\label{fig:noisy_compressible}
\end{figure}

In Figure~\ref{fig:noisy_compressible} we plot three averaged performance metrics (vertical axes) against $\tau \in \{10^7, 10^8, \ldots, 10^{16}\}$ (horizontal axes) for each of $\ell = 1, 2$, and $3$.  The plotted metrics are:
\begin{enumerate}
    \item \textbf{Measurement error} (left):
$$\frac{1}{100}\sum_{j=1}^{100}
\min\left\{\left\|\tilde{\bf u}_{j,\ell}
- M{\bf u}_{j,\ell}\right\|_2,\left\|\tilde{\bf u}_{j,\ell}
+ M{\bf u}_{j,\ell}\right\|_2\right\}.$$
\item \textbf{Compressive Sensing (CS) Recovery error} (center):
$$\frac{1}{100}\sum_{j=1}^{100}
\min\left\{\left\|{\bf u}_{j,\ell}
- \mathcal{A}\left(\tilde{\bf u}_{j,\ell}\right)\right\|_2,\left\|{\bf u}_{j,\ell}
+ \mathcal{A}\left(\tilde{\bf u}_{j,\ell}\right)\right\|_2\right\}.$$
\item \textbf{Support recovery fraction} (right):
$$\frac{1}{100} \sum_{j=1}^{100}\frac{\left|\operatorname{supp}_{2s}\left(
\mathcal{A}(\tilde{\bf u}_{j,\ell})\right)
\cap\, S_{j,\ell}\right|}{s},$$
where $\operatorname{supp}_{2s}({\bf x})$  denotes the indices of the $2s$ largest entries of ${\bf x}$ in magnitude.
\end{enumerate}
Figure~\ref{fig:noisy_compressible} demonstrates that the $MAM^*$ approach is not limited to exactly sparse eigenvectors associated with exponentially decaying eigenvalues.  The method still accurately approximates the original ${\bf u}_{j,\ell}$-vectors across a wide range of noise strengths (i.e., $\tau$ values) despite them being compressible and associated with linearly decaying eigenvalues.  As one can see, the relative measurement and recovery errors remain stable for $\tau$ up to approximately $10^{14}$ after which the additive noise matrix $\tau {\bf g}_j {\bf g}_j^*$ completely overwhelms the compressible components.

\section{Proofs of Main Results}
\label{sec:MainResFinalProofs}

We prove our two example main results below.

\subsection{Proof of Theorem~\ref{THM:MAINRESULT_LINEARcosamp}}
\label{sec:PROOFOFMAINRESULT_LINEAR}

We reproduce Theorem~\ref{THM:MAINRESULT_LINEARcosamp}
below for ease of reference.  

\begin{theorem} 
Let $q \in (0,1/3)$, $c \in [1,\infty)$, $\ell, r \in [N]$, and $\epsilon \in (0,1)$ be such that $\epsilon < \min \left\{ \frac{1}{20}, \frac{1}{4}\left(\frac{1-3q}{1+q}\right) \right\}$ and $2 \leq \ell \leq r$. Suppose that $\A\in \mathbbm{C}^{\N\times\N}$ is Hermitian and PSD with eigenvalues $\lambda_1 \geq \lambda_2 \geq \dots \geq \lambda_N \geq 0$ satisfying
\begin{enumerate}
    \item $\lambda_j = c q^j$ for all $j \in [\ell] \subseteq [r]$, and
    \item $\|\A_{\backslash \r}\|_* \leq \epsilon \lambda_\ell$.
\end{enumerate}
Choose $s\in [\N]$, $p, \eta \in (0,1)$, and form a random matrix $\M \in\mathbbm{C}^{\m\times\N}$ with $m = \mathcal{O}\left( \max \left\{s,\frac{r}{\epsilon^2} \right\} \log^4\left( N/p\epsilon^2 \right) \right)$ as per Theorem~\ref{Thm:MainCosampMatrixSetup}.  Let $\utilde_j$ and $\u_j$ be the ordered eigenvectors of $MAM^*$ \eqref{def:SVD_A_Atilde} and $A$ \eqref{def:SVD_A}, respectively, for all $j \in [\ell]$.  Then, there exists a compressive sensing algorithm $\mathcal{A}_{\rm lin}: \mathbbm{C}^m \rightarrow \mathbbm{C}^N$ and an absolute constant $c' \in \mathbbm{R}^+$ such that 
\begin{align*}
		&\min_{\phi \in [0,2\pi)} \left \| \mathbbm{e}^{\mathbbm{i} \phi} \u_j - \mathcal{A}_{\rm lin}(\utilde_j) \right \|_2 < c' \cdot \max \left\{ \eta, \frac{1}{\sqrt{s}} \| \u_j - (\u_j)_s\|_1 + \sqrt{\epsilon} \cdot q^{1-j} \right\}
\end{align*}
holds for all $j \in [\ell]$ with probability at least $1-p$.  Furthermore, all $\ell$ estimates $\left\{ \mathcal{A}_{\rm lin}(\utilde_j) \right \}_{j \in [\ell]}$ can be computed in $\mathcal{O}\left(m^3 + \ell N \log N \cdot \log(1/\eta) \right)$-time from $MAM^* \in \mathbbm{C}^{m \times m}$.
\end{theorem}

\begin{proof} 
Define $\mathcal{A}_{\rm lin}(\cdot) := g ( \cdot )$ as per Theorem~\ref{Thm:CoSamp} and let $\phi'_j \in [0,2\pi)$ be such that $\left\| \utilde_j - \mathbbm{e}^{\mathbbm{i} \phi'_j} M \u_j \right\|_2 = \min_{\phi \in [0,2\pi)} \| \mathbbm{e}^{\mathbbm{i} \phi} \utilde_j - M \u_j \|_2$ for all $j \in [\ell]$.  Theorem~\ref{Thm:MainCosampMatrixSetup} implies that, with probability at least $1-p$, $\M \in\mathbbm{R}^{\m\times\N}$ will satisfy the assumptions in Theorem~\ref{theorem:MAM}. As a consequence, Theorem~\ref{MainTHM:ExpDecayingSingVlaues} tells us that 
\begin{align}\label{equ:MainTh1proof}
    \left\| \utilde_j - \mathbbm{e}^{\mathbbm{i} \phi'_j} M \u_j \right\|_2 < 7 \sqrt{\epsilon} \cdot q^{1-j} ~~\forall j \in [\ell].
\end{align}
In addition, whenever \eqref{equ:MainTh1proof} holds part $(ii)$ of Theorem~\ref{Thm:MainCosampMatrixSetup} also holding then further implies that
\begin{align*}
\min_{\phi \in [0,2\pi)} \big \| \mathbbm{e}^{\mathbbm{i} \phi} \u_j -& \mathcal{A}_{\rm lin}(\utilde_j) \big \|_2 \leq \left\| \mathbbm{e}^{\mathbbm{i} \phi'_j} \u_j - \mathcal{A}_{\rm lin}(\utilde_j) \right\|_2\\
		&\qquad = \left \| \mathbbm{e}^{\mathbbm{i} \phi'_j} \u_j - g \left ( \mathbbm{e}^{\mathbbm{i} \phi'_j} M \u_j + \left( \utilde_j - \mathbbm{e}^{\mathbbm{i} \phi'_j} M \u_j \right) \right ) \right \|_2 \\ & \qquad  \leq c'' \cdot \max \left\{ \eta, \frac{1}{\sqrt{s}} \| {\u}_j - (\u_j)_s\|_1 + \left \| \utilde_j - \mathbbm{e}^{\mathbbm{i} \phi'_j} M \u_j \right \|_2 \right \}\\ 
        & \qquad  < c'' \cdot \max \left \{ \eta, \frac{1}{\sqrt{s}} \| \u_j - (\u_j)_s \|_1 + 7 \sqrt{\epsilon} \cdot q^{1-j} \right \}
\end{align*}
holds for all $j \in [\ell]$, where $c''$ is the absolute from the error bound in part $(ii)$ of Theorem~\ref{Thm:MainCosampMatrixSetup}.  Setting $c' = 7c''$ finishes the proof of the error bound.

Turning our attention to the claimed measurement and runtime bounds, we refer the reader to Theorem~\ref{Thm:MainCosampMatrixSetup} for the quoted upper bound on $m$.  Concerning the time required to compute $\left\{ \mathcal{A}_{\rm lin}(\utilde_j) \right \}_{j \in [\ell]}$, we note that the top $\ell \leq m$ eigenvectors $\left\{ \utilde_j \right \}_{j \in [\ell]}$ of $MAM^*$ can be computed in $\mathcal{O}(m^3)$-time (see, e.g., \cite{golub1996matrix}).  Furthermore, once those eigenvectors are in hand, $\left\{ \mathcal{A}_{\rm lin}(\utilde_j) \right \}_{j \in [\ell]}$ can then be computed in $\mathcal{O}(\ell N \log N \cdot \log(1/\eta))$-time since $g: \mathbbm{C}^{m} \rightarrow \mathbbm{C}^N$, can always be evaluated in $\mathcal{O}(N \log N \cdot \log(1/\eta))$-time by Theorem~\ref{Thm:MainCosampMatrixSetup}. \end{proof}

\subsection{Proof of Theorem~\ref{THM:MAINRESULT_SUBLINEAR}}
\label{sec:PROOFOFMAINRESULT_SUBLINEAR}

We reproduce Theorem~\ref{THM:MAINRESULT_SUBLINEAR}
below for ease of reference.  

\begin{theorem} 
Let $q \in (0,1/3)$, $c \in [1,\infty)$, and $\ell, r, 1/\epsilon \in [N]$ be such that $\epsilon < \min \left\{ \frac{1}{20}, \frac{1}{4}\left(\frac{1-3q}{1+q}\right) \right\}$ and $2 \leq \ell \leq r$. Suppose that $\A\in \mathbbm{C}^{\N\times\N}$ is Hermitian and PSD with eigenvalues $\lambda_1 \geq \lambda_2 \geq \dots \geq \lambda_N \geq 0$ satisfying
\begin{enumerate}
    \item $\lambda_j = c q^j$ for all $j \in [\ell] \subseteq [r]$, and
    \item $\|\A_{\backslash \r}\|_* \leq \epsilon \lambda_\ell$.
\end{enumerate}
Choose $s\in [\N]$, $p \in (0,1)$, and form a random matrix $\M \in\mathbbm{C}^{\m\times\N}$ with $m = \mathcal{O}\left( \max \left\{ s^2, \frac{r^2}{\epsilon^2}  \right\} \log^5(N/p) \right)$ as per Theorem~\ref{Thm:MainSublinearMatrixSetup}.  Let $\utilde_j$ and $\u_j$ be the ordered eigenvectors of $MAM^*$ \eqref{def:SVD_A_Atilde} and $A$ \eqref{def:SVD_A}, respectively, for all $j \in [\ell]$.  Then, there exists a compressive sensing algorithm $\mathcal{A}_{\rm sub}: \mathbbm{C}^m \rightarrow \mathbbm{C}^N$ and $\beta^{A}_{M} \in \mathbbm{R}^+$ such that 
\begin{align*}
		&\min_{\phi \in [0,2\pi)} \left \| \mathbbm{e}^{\mathbbm{i} \phi} \u_j - \mathcal{A}_{\rm sub}(\utilde_j) \right \|_2\\ 
        & \qquad \qquad < \| \u_j - (\u_j)_{ 2s } \|_2 + 6(1+\sqrt{2}) \left( \frac{\| \u_j - (\u_j)_s \|_1}{\sqrt{ s }} + \beta^{A}_{M} \sqrt{\epsilon} \cdot q^{1-j}\right), 
\end{align*}
holds for all $j \in [\ell]$ with probability at least $1-p$.  Furthermore, all $\ell$ estimates $\left\{ \mathcal{A}_{\rm sub}(\utilde_j) \right \}_{j \in [\ell]}$ can be computed in just $\mathcal{O}\left(m^3 \right)$-time from $MAM^* \in \mathbbm{C}^{m \times m}$.
\end{theorem}

\begin{proof} 
Define $\mathcal{A}_{\rm sub}(\cdot) := \sqrt { K(1 + \lceil \log_2 N \rceil) } D f ( \cdot )$ as per Theorem~\ref{Thm:MainSublinearMatrixSetup} and let $\phi'_j \in [0,2\pi)$ be such that $\left\| \utilde_j - \mathbbm{e}^{\mathbbm{i} \phi'_j} M \u_j \right\|_2 = \min_{\phi \in [0,2\pi)} \| \mathbbm{e}^{\mathbbm{i} \phi} \utilde_j - M \u_j \|_2$ for all $j \in [\ell]$.  Part $(i)$ of Theorem~\ref{Thm:MainSublinearMatrixSetup} implies that $\M \in\mathbbm{R}^{\m\times\N}$ will satisfy the assumptions in Theorem~\ref{theorem:MAM} with probability at least $1-p$.  As a consequence, Theorem~\ref{MainTHM:ExpDecayingSingVlaues} tells us that 
\begin{align}\label{equ:MainTh2proof}
    \left\| \utilde_j - \mathbbm{e}^{\mathbbm{i} \phi'_j} M \u_j \right\|_2 < 7 \sqrt{\epsilon} \cdot q^{1-j} ~~\forall j \in [\ell]
\end{align}
holds with probability at least $1-p$.  In addition, whenever \eqref{equ:MainTh2proof} holds part $(ii)$ of Theorem~\ref{Thm:MainSublinearMatrixSetup} then further implies that 
\begin{align*}
&\min_{\phi \in [0,2\pi)} \left \| \mathbbm{e}^{\mathbbm{i} \phi} \u_j - \mathcal{A}_{\rm sub}(\utilde_j) \right \|_2 \leq \left\| \mathbbm{e}^{\mathbbm{i} \phi'_j} \u_j - \mathcal{A}_{\rm sub}(\utilde_j) \right\|_2\\
		&\qquad = \left \| \mathbbm{e}^{\mathbbm{i} \phi'_j} \u_j - \sqrt { K(1 + \lceil \log_2 N \rceil) } \cdot D f \left ( \mathbbm{e}^{\mathbbm{i} \phi'_j} M \u_j + \left( \utilde_j - \mathbbm{e}^{\mathbbm{i} \phi'_j} M \u_j \right) \right ) \right \|_2 \\ & \qquad  \leq \| \u_j - (\u_j)_{ 2s } \|_2 + 6(1+\sqrt{2}) \left( \frac{\| \u_j - (\u_j)_s \|_1}{\sqrt{ s }} + \beta_{m}\left(\utilde_j - \mathbbm{e}^{\mathbbm{i} \phi'_j} M \u_j \right) \left\| \utilde_j - \mathbbm{e}^{\mathbbm{i} \phi'_j} M \u_j \right\|_2 \right) \\ &
        \qquad < \| \u_j - (\u_j)_{ 2s } \|_2 + 6(1+\sqrt{2}) \left( \frac{\| \u_j - (\u_j)_{ s } \|_1}{\sqrt{ s }} + 7 \beta_{m}\left(\utilde_j - \mathbbm{e}^{\mathbbm{i} \phi'_j} M \u_j \right) \sqrt{\epsilon} \cdot q^{1-j} \right)
\end{align*}
holds for all $j \in [\ell]$.  Setting $\beta^{A}_{M} := 7 \max_{j \in [\ell]} \beta_{m}\left(\utilde_j - \mathbbm{e}^{\mathbbm{i} \phi'_j} M \u_j \right)$ where $\beta_{m}(\cdot)$ is defined as per \eqref{equ:DefofBeta_mn} now finishes the proof of the error bound.

Turning our attention to the claimed measurement and runtime bounds, we refer the reader to Remark~\ref{remark:sublinearmbound} for the quoted upper bound on $m$.  Concerning the time required to compute $\left\{ \mathcal{A}_{\rm sub}(\utilde_j) \right \}_{j \in [\ell]}$, we note that the top $\ell \leq m$ eigenvectors $\left\{ \utilde_j \right \}_{j \in [\ell]}$ of $MAM^*$ can be computed in $\mathcal{O}(m^3)$-time (see, e.g., \cite{golub1996matrix}).  Furthermore, once those eigenvectors are in hand, $\left\{ \mathcal{A}_{\rm sub}(\utilde_j) \right \}_{j \in [\ell]}$ can then be computed in $\mathcal{O}(m^2)$-time since $f: \mathbbm{C}^{m} \rightarrow \mathbbm{C}^N$, which outputs $2s$-sparse vectors, can always be evaluated in $\mathcal{O}\left(m \right)$-time (recall Theorem~\ref{Thm:FinalsublinearCSRes}).
\end{proof}

\section{Conclusion}
\label{sec:Conclusion}

Making use of compressive sensing and sketching principles, we proposed the $MAM^*$ method (Algorithm~\ref{Alg:ApproxfromMAM}) for the sparse approximation of eigenvectors of very large-scale, approximately low-rank matrices. More specifically, we studied two variants of $MAM^*$: a first one based on CoSaMP reconstruction and a second one relying on sublinear-time recovery. The CoSaMP-based variant of $MAM^*$, analyzed in Theorem~\ref{THM:MAINRESULT_LINEARcosamp}, can compute accurate $s$-sparse approximations to the first $\ell$ eigenvectors of an approximately rank-$r$, Hermitian PSD matrix $A\in \mathbbm{C}^{N \times N}$ with high probability given $MAM^*$ as input, in $\mathcal{O}((\max\{s,r\}^3 + \ell N) \cdot \text{polylog}(N,s))$-time and using $\mathcal{O}(N)$-memory. On the other hand, the sublinear-time variant studied in  Theorem~\ref{THM:MAINRESULT_SUBLINEAR} can achieve the same task in $\mathcal{O}(\max\{s, r\}^6 \cdot \text{polylog}(N))$-time (and memory). We also numerically evaluated the $MAM^*$ method by implementing and testing its, e.g., sublinear-time variant on $(2^{27}-1) \times (2^{27}-1) \approx 10^8 \times 10^8$ matrices for levels of sparsity and rank up to $s=100$ and $r=50$, respectively. 

We conclude by discussing some of the limitations of our study and pointing to potential avenues of future work. First, throughout this work we assumed the matrix $A \in \mathbbm{C}^{N \times N}$ to be  Hermitian and PSD. This assumption was made primarily to keep the technical level of the analysis and presentation relatively moderate. However, we don't consider this restriction to be necessary and we think that the non-Hermitian case could be addressed via asymmetric sketches of the form $MAN^*$. Second, other linear-time variants of $MAM^*$ could be easily constructed by considering other compressive sensing reconstruction methods beyond CoSaMP such as, e.g., Orthogonal Matching Pursuit (OMP) or Iterative Hard Thresholding (IHT). The corresponding version of Theorem~\ref{THM:MAINRESULT_LINEARcosamp} would be then obtained by combining the same proof strategy of Section~\ref{sec:PROOFOFMAINRESULT_LINEAR} with a sparse recovery guarantee for the reconstruction algorithm of choice. Third, some aspects of our theoretical analysis deserve further investigation. For example, while a numerical study of the constant $\beta_M^A$ appearing in Theorem~\ref{THM:MAINRESULT_SUBLINEAR} shows that it has moderate size (see Section~\ref{sec:BetaExperiments}), it would be useful to derive a theoretical upper bound for it under suitable assumptions (see also Remark~\ref{remark:sublinearBetaissmall}). Another theoretical issue worth investigating is the relaxation of the geometric decay assumption on the eigenvalues of $A$ in Theorems~\ref{THM:MAINRESULT_LINEARcosamp} and \ref{THM:MAINRESULT_SUBLINEAR}. Finally, to assess the real potential of $MAM^*$ further numerical experimentation is certainly required. Applications to, e.g., sparse PCA with stupendously large real-world datasets and to the numerical solution of high-dimensional PDE-based eigenvalue problems are two promising research directions currently under investigation.

\section*{Acknowledgements}

Mark Iwen and Edem Boahen were supported in part by NSF DMS 2106472. Simone Brugiapaglia was partially supported by the Natural Sciences and Engineering Research Council of
Canada through grant RGPIN-2020-06766 and the Fonds de Recherche du Qu\'ebec
Nature et Technologies through grants 313276 and 359708. Simone Brugiapaglia and Mark Iwen would like to thank Craig Gross for insightful discussions in the initial phase of this project.

\bibliography{biblioSOda1}{}
\bibliographystyle{abbrv}

\end{document}